\documentclass[preprint,11pt]{article}

\usepackage{amssymb,amsfonts,amsmath,amsthm,amscd,dsfont,mathrsfs}
\usepackage{eqsection,graphicx,float,psfrag,epsfig,color}

\newtheorem{claim}{Claim}
\newtheorem{lemma}{Lemma}

\newtheorem{theorem}{Theorem}

\newtheorem{corollary}[claim]{Corollary}
\newtheorem{definition}[claim]{Definition}

\newcommand{\enp} {}

\newcommand{\p}{\partial}

\newcommand{\f}{\frac}

\newcommand{\salf}[2]{\widehat{\alpha}_{#1\backslash #2}}
\newcommand{\sbet}[2]{\widehat{\beta}_{#1\backslash #2}}
\newcommand{\alf}[2]{\alpha_{#1\backslash #2}}
\newcommand{\mass}[2]{\rho_{#1\backslash #2}}
\newcommand{\imass}[2]{\widetilde{\rho}_{#1\backslash #2}}

\newcommand{\sDel}[2]{\widehat{\Delta}_{#1\backslash #2}}
\newcommand{\bound}[2]{A_{#1\backslash #2}}
\newcommand{\tightbound}[2]{B_{#1\backslash #2}}
\newcommand{\tbound}[2]{\widetilde{A}_{#1\backslash #2}}
\newcommand{\bet}[2]{\beta_{#1\backslash #2}}
\newcommand{\soff}[2]{\widehat{m}_{#1\rightarrow #2}}
\newcommand{\off}[2]{m_{#1\rightarrow #2}}

\newcommand{\barbound}{\underline{A}}
\newcommand{\bartightbound}{\underline{B}}
\newcommand{\tbarbound}{\widetilde{\underline{A}}}
\newcommand{\baralf}{\underline{\alpha}}
\newcommand{\barbet}{\underline{\beta}}
\newcommand{\sbaralf}{\underline{\widehat{\alpha}}}
\newcommand{\sbarbet}{\underline{\widehat{\beta}}}
\newcommand{\bargamma}{\underline{\gamma}}

\newcommand{\baroff}{\underline{m}}
\newcommand{\barsoff}{\underline{\widehat{m}}}

\newcommand{\barmass}{\underline{\rho}}
\newcommand{\sbarDelta}{\widehat{\underline{\Delta}}}

\def\Times{{\mathfrak T}}
\def\reals{{\mathds R}}
\def\naturals{{\mathds N}}
\def\di{{\partial i}}
\def\cC{{\cal C}}
\def\cG{{\cal G}}
\def\cF{{\cal F}}

\def\eps{{\epsilon}}
\def\c{\textup{c}}
\def\s{\textup{s}}
\def\b{\textup{b}}
\def\L{\textup{\tiny L}}
\def\R{\textup{\tiny R}}
\def\P{\textup{\tiny P}}
\def\RW{\textup{\tiny RW}}
\def\top{\textup{top}}
\def\bot{\textup{bot}}
\def\up{\textup{up}}

\def\boundary{b}
\def\damp{\kappa}
\def\Bo{{\mathds B}}
\def\dj{{\partial j}}
\def\u0{\underline{0}}
\def\ux{\underline{x}}
\def\parti{i_\P}
\def\NB{{\textup {\tiny NB}}}
\def\len{{\ell}}

\definecolor{Red}{rgb}{1,0,0}
\definecolor{Blue}{rgb}{0,0,1}
\definecolor{Olive}{rgb}{0.41,0.55,0.13}
\definecolor{Green}{rgb}{0,1,0}
\definecolor{MGreen}{rgb}{0,0.8,0}
\definecolor{DGreen}{rgb}{0,0.55,0}
\definecolor{Yellow}{rgb}{1,1,0}
\definecolor{Cyan}{rgb}{0,1,1}
\definecolor{Magenta}{rgb}{1,0,1}
\definecolor{Orange}{rgb}{1,.5,0}
\definecolor{Violet}{rgb}{.5,0,.5}
\definecolor{Purple}{rgb}{.75,0,.25}
\definecolor{Brown}{rgb}{.75,.5,.25}
\definecolor{Grey}{rgb}{.5,.5,.5}
\definecolor{Pink}{rgb}{1,0,1}
\definecolor{DBrown}{rgb}{.5,.34,.16}

\definecolor{Black}{rgb}{0,0,0}

\begin{document}

\begin{titlepage}

\title{A Natural Dynamics for Bargaining on Exchange Networks}

\author{Yashodhan Kanoria\thanks{Department of Electrical Engineering,
Stanford University},\;\;
Mohsen Bayati${}^*$,\\
\\
Christian Borgs\thanks{Microsoft Research New England},\;\;
Jennifer Chayes${}^{\dagger}$\;
and Andrea Montanari${}^{*,}$\thanks{Department of Statistics,
Stanford University}}

\date{}

\maketitle
\thispagestyle{empty}

\begin{abstract}
Bargaining networks model the behavior of a set of players
that need to reach pairwise agreements for making profits.
Nash bargaining solutions are special outcomes of such
games that are both stable and balanced. Kleinberg
and Tardos proved a sharp algorithmic characterization of
such outcomes, but left open the problem of how the actual
bargaining process converges  to them.
A partial answer was provided by Azar et al. who
proposed a distributed algorithm for constructing
Nash bargaining solutions, but without polynomial bounds on its
convergence rate.
In this paper, we introduce a simple and natural model for this process, and
study its convergence rate to Nash bargaining solutions.

At each time step, each player proposes a deal to each
of her neighbors. The proposal consists of a share of
the potential profit in case of agreement. The share is
chosen to be balanced in Nash's sense
as far as this is feasible (with respect
to the current best alternatives for both players).
We prove that, whenever the Nash bargaining solution
is unique (and satisfies a positive gap condition)
this dynamics converges to it in polynomial time.

Our analysis is based on an approximate \emph{decoupling}
phenomenon between the dynamics on different substructures of the
network. This approach may be of general interest for the
analysis of local algorithms on networks.
\end{abstract}

\end{titlepage}

\section{Introduction and main results}

Exchange networks model social and economic relations among
individuals under the premise that any relationship has a potential
value for its partners. In a purely economic setting,
one can imagine that each relation corresponds to a trading
opportunity, and its value is the amount of money to be
earned from the trade.
A fascinating question in this context is that of
how network structure influences the power balance between
nodes (i.e. their earnings).

Controlled experiments \cite{NET,Lucas,Skvoretz} have been carried out by
sociologists in a set-up that can be summarized as follows.
A graph $G=(V,E)$ is defined, with positive weights
$w_{ij}>0$ associated  to the edges $(i,j)\in E$.
A player sits at each node of this network,
and two players connected by edge $(i,j)$ can
share a profit of $w_{ij}$ dollars if they agree to
trade with each other.
Each player can trade with at most one of her neighbors
(this is called the \emph{$1$-exchange rule}), so that
a set of valid trading pairs forms a matching $M$ in the graph $G$.
It is often the case that players are provided information
only about their immediate neighbors.

Network exchange theory studies the possible outcomes of such a process.
While each instance admits a multitude of outcomes,
special classes of outcomes are selected on the basis
of `desirable' properties. In this paper, we focus on
`balanced outcomes', a solution concept that dates back to Nash's bargaining
theory \cite{Nash}, and was generalized in \cite{Rochford,CookY,KT}.

A balanced outcome, or Nash bargaining (NB) solution, is a pair
$(M,\bargamma)$ where $M\subseteq E$ is a matching of $G$, and
$\bargamma = \{\gamma_i\, :\, i\in V\}$ is the vector of players'
profits.
Clearly $\gamma_i\ge 0$, and
$(i,j) \in M$ implies $\gamma_i + \gamma_j = w_{ij}$.
Denote by $\partial i$ the set of neighbors of node $i$ in $G$.
The pair $(M,\gamma)$  is a NB solution if it satisfies the following
requirements.

\noindent\textit{Stability.} If player $i$ is trading with $j$, then
she cannot earn more by simply changing her trading partner.
Formally $\gamma_i + \gamma_j \ge w_{ij}$ for all 
$(i,j) \in E\setminus M$.

\noindent\textit{Balance.} If player $i$ is trading with $j$, then the surplus of $i$ over his best alternative must be
equal to the surplus of $j$ over his best alternative. Mathematically,
\begin{align}
\gamma_i - \max_{k \in \di \backslash j} (w_{ik}- \gamma_k)_+ =
\gamma_j - \max_{l \in \dj \backslash i} (w_{jl}- \gamma_l)_+
\label{eq:balance}
\end{align}
for all $(i,j) \in M$.

It turns out that the interplay between the $1$-exchange rule
and the stability and balance conditions results in highly non-trivial
predictions regarding the influence of network structure on
individual earnings. Some of these predictions agree with experimental 
findings, but alternative predictive frameworks exist as well \cite{Skvoretz}.
%
%
\subsection{A natural dynamics}

It is a fundamental open question whether NB solutions
describe the outcomes of actual bargaining processes.
The stream of controlled experiments on small networks
will surely help to get an answer \cite{NET}.
On the other hand, an important step
forward was achieved by Kleinberg and Tardos \cite{KT} who proved
that NB solutions can be constructed in polynomial time.

However, even a superficial look at experimental conditions
reveals that players  cannot possibly run the algorithm described
in \cite{KT}.
There are two possibilities: Either there exists a realistic model
for the bargaining dynamics that converges to NB solutions,
or the solution concept has to be revised.
For the former possibility, the underlying dynamics should satisfy the following requirements:
$(1)$ It should converge to NB solutions in polynomial time;
$(2)$ It should be \emph{natural}.

While the first requirement is easy to define and motivate, the
second one is more subtle but not less important. A few properties
of a natural dynamics are the following ones: It should be
\emph{local},
i.e. involve limited information exchange along edges and processing
at nodes; It should be \emph{time invariant}, i.e. the players' behavior
should be the same/similar on identical local information at different times;
It should be \emph{interpretable}, i.e. the information  exchanged
along the edges should have a meaning for the players involved, and
should be consistent with reasonable behavior for players.

In the model we propose, at each time $t$, each player sends a message to
each of her neighbors. The message has the meaning of `best current
alternative'. We denote the message from player $i$ to player $j$
by $\alf{i}{j}^t$. Player $i$ is telling player $j$ that she (player $i$)
can currently earn $\alf{i}{j}^t$ elsewhere,
if she chooses not to trade with $j$.

The vector of all such messages is denoted by $\baralf^t \in \reals_+^{2|E|}$.
Each agent $i$ makes an `offer' to each of
her neighbors, based on her own `best alternative' and that of her neighbor. The offer from node $i$ to $j$ is denoted by $\off{i}{j}^t$
and computed according to
\vspace{-0.5cm}

\begin{align}
	\off{i}{j}^t =
 (w_{ij}-\alf{i}{j}^t)_+ - \frac{1}{2}(w_{ij}-\alf{i}{j}^t-\alf{j}{i}^t)_+\, .
	\label{eq:off_def}
\end{align}
\vspace{-0.5cm}

It is easy to realize that this definition corresponds to the following
policy: $(i)$ An offer is always non-negative, and a positive offer is never
larger than $w_{ij}-\alf{i}{j}^t$ (no player is interested in earning less
than what is currently being offered);
$(ii)$ Subject to the above constraints, the surplus
$(w_{ij}-\alf{i}{j}^t-\alf{j}{i}^t)$ (if non-negative) is shared equally.
We denote by $\baroff^t\in\reals_+^{2|E|}$ the vector of offers.

Notice that $\baroff^t$ is just a deterministic coordinate-by-coordinate
function of $\baralf^t$. In the rest of the paper we shall describe
the network status uniquely through the latter vector,  and use $\baroff|_{\baralf^t}$
to avoid ambiguity when required.
Each node can estimate its potential
earning based on the network status, using
\begin{align}
\gamma_i^t \equiv \max_{k \in \di} \, \off{k}{i}^t,\label{eq:bargamma}
\end{align}
the corresponding vector being denoted by $\bargamma^t\in\reals_+^{2|E|}$.

Messages are updated synchronously through the network,
according to the rule
\begin{align}
\alf{i}{j}^{t+1}= \damp \max_{k \in \di \backslash j} \off{k}{i}^{t}+(1-\damp)\,\alf{i}{j}^{t}\, .
\label{eq:update}
\end{align}
Here $\damp \in (0,1)$ is a `damping' factor: $(1-\damp)$ can be thought of as
the inertia on the part of the nodes to update their outgoing messages.
The use of  $\damp<1$ eliminates pathological behaviors
related to synchronous updates (such as oscillations on even-length
cycles). We expect that the use asynchronous updates also eliminates
such problems.

Throughout the paper we let $W\equiv\max_{(ij)\in E}w_{ij}$.
It is easy to see that this implies
$\baralf^{t} \in [0,W]^{2|E|}$, $\baroff^t \in [0,W]^{2|E|}$
and $\bargamma^t \in [0,W]^{|V|}$ at all times
(unless the initial condition violates
this bounds). Thus we call $\baralf$ a `valid'
message vector if $\baralf \in [0,W]^{2|E|}$.
%
%
\subsection{Main results: Fixed point properties and convergence}

Our first result is that fixed points of
the update equations (\ref{eq:off_def}), (\ref{eq:update})
(hereafter referred to as `natural dynamics')
are indeed in correspondence with Nash bargaining solutions 
when such solutions exist.
Recall the LP relaxation to the maximum weight matching problem
\begin{eqnarray}
\textup{maximize} && \sum_{(i,j) \in E} w_{ij} x_{ij},\nonumber\\
\textup{subject to}&&
\sum_{j\in \di} x_{ij} \le 1 \;\;\; \forall i \in V,\;\;\;\;\;\;\;
x_{ij}\ge 0\;\;\;\forall (i,j) \in E
\label{prob:mwm_relaxation}
\end{eqnarray}
\begin{theorem}\label{thm:fp_dualopt}
Let $G$ be an instance for which the LP (\ref{prob:mwm_relaxation})
admits a unique optimum, and this is integer.
Let $(\baralf,\baroff,\bargamma)$ be a fixed point of
the natural dynamics. Then  $\bargamma$ is the allocation
of a Nash bargaining solution (i.e. there exists a matching $M$,
such that the pair $(M,\gamma)$ is stable and balanced).
Conversely, every Nash bargaining solution $(M,\bargamma_{\NB})$,
corresponds to a unique
fixed point of the natural dynamics with $\bargamma=\bargamma_{\NB}$.
\end{theorem}

The natural dynamics appear to converge rapidly to
a fixed point on all the cases we studied.
In particular, we will prove this to be the case whenever
the NB solution is unique, and under the assumption
of a \emph{positive gap $\sigma>0$}.
The definition of `gap' is somewhat technical and is deferred to
Section \ref{sec:KT+gap}. There we will also argue that the conditions
is generic. Further, thanks to \cite{KT}, it can be checked efficiently.
In the rest of the paper we use $C$ to represent any constant that is
independent of the instance $G$, and $n, \sigma$ in particular.
\begin{theorem}\label{thm:Convergence}
Let $G$ be an instance having unique Nash bargaining solution with
gap $\sigma>0$, and let $\bargamma_{\NB}$ denote the corresponding allocation.
Then there exists
\vspace{-.2cm}
\begin{eqnarray}
T_*(n,\sigma,\eps)= C\, n^7
\left(\frac{W}{\sigma}\, +\log \Big(1+\frac{\sigma}{\eps}\Big)
\right)\, ,
\vspace{-.2cm}
\end{eqnarray}
such that, for any initial condition with
$\baralf^0\in[0,W]^{2|E|}$, and any $t\ge T_*$
the natural dynamics yields earning estimates $\bargamma^t$, with
$|\gamma^t_i-\gamma_{\NB,i}|\le \eps$ for all $i\in V$.
\end{theorem}
It is worth stressing that the assumption of unique NB solution
seems to be a weakness of our proof technique, rather than a
necessary condition for fast convergence. In Section \ref{sec:Bipartite}
we show indeed that the natural dynamics \emph{always converges}
on bipartite graphs if run from extremal initial conditions.
On the other hand, the class of instances with unique NB solution
already includes a large class of cases.

The previous theorem provides a convergence guarantee for
the earnings $\bargamma^t$ that players expect.
Our last result shows the correct
pairing among players also emerges from the dynamics.
\begin{theorem}
Under the hypotheses of Theorem \ref{thm:Convergence}, assume
$t\ge T_*(n,\sigma,\sigma/3)$. Then, for each node
$i\in V$ receiving non-zero offers,
there exists a unique neighbor $P(i)\in \di$
such that the offer $\off{P(i)}{i}^i$ from $P(i)$ to $i$
is strictly larger than the offers $\off{l}{i}^t$ from other nodes
$l\in \di\setminus P(i)$.

Further $j=P(i)$ if and only if $i=P(j)$.
The pairs  $(i,P(i))$ thus defined coincide with the
ones in the unique Nash bargaining solution.
\end{theorem}
This theorem follows immediately from the proof
of Theorem \ref{thm:Convergence}, as we will prove there that, for
$t\ge T_*(n,\sigma,\eps)$,  $|\alf{i}{j}^t-\alf{i}{j}|\le \eps$, where $\baralf$ is the
unique fixed point.
For $\eps\le\sigma/3$, this is sufficient to unambiguously determine pairings.

%
%
\subsection{Related work}

Following \cite{Rochford,CookY}, Kleinberg and Tardos
\cite{KT} first considered balanced
outcomes on general exchange networks and proved
that a network $G$ admits a balanced outcome if and only if
it admits a stable outcome. Further, the latter happens
if and only if a linear programming relaxation of the maximum
weight matching problem on $G$ admits an integral optimum.

The same paper describes a polynomial algorithm for
constructing balanced outcomes. This is in turn based on
the dynamic programming algorithm of
Aspvall and Shiloach \cite{Aspvall}
for solving systems of linear inequalities.
Our convergence proof exploits the structural decomposition
of the network that is produced by this algorithm.

Alternative solution concepts for bargaining on networks
were studied in \cite{Kearns1}.

Azar and co-authors \cite{Azar} first studied the question
as to whether a balanced outcome can be produced by a local dynamics,
and were able to answer positively. Their results left
however two outstanding challenges:
$(I)$ The bound on the convergence time proved in \cite{Azar}
is exponential in the network size, and therefore does not
provide a solid justification for convergence to NB solutions in large
networks; $(II)$ The algorithm analyzed by these authors
first selects a matching $M$ in $G$ \cite{Bayati}, corresponding to the
pairing of players that trade.
In a second phase the algorithm determines the profit of each player.
While such an algorithm
can be implemented in a distributed way, Azar et al. point out
that it is not entirely realistic. Indeed the rules of the dynamics change
after the matching is found. Further, if the pairing
is established at the outset, the players lose their
bargaining power.

The present paper aims at tackling these challenges.
%
%
\subsection{Outline of the proof}

We next describe the main steps in the proof of our convergence
result, Theorem \ref{thm:Convergence}. This is based on the following strategy:
\vspace{0.07cm}

\noindent {\bf 1.} Prove that, under the positive gap condition, the
natural dynamics on different substructures of the networks approximately
\emph{decouples}.

\noindent {\bf 2.}
Analyze the dynamics on each structure by comparison with an
appropriate random walk process.

\vspace{0.07cm}

An approach based on these two steps might be applicable to the analysis of a
wide class of local algorithms on networks.

Step 1 above requires recalling
the construction in \cite{KT}, which we do next.
Section \ref{sub:Convergence} describes the main steps of the proof.
The fixed point properties of the natural dynamics are summarized
in Section \ref{sec:FixedPoint}, which outlines the  proof of
Theorem \ref{thm:fp_dualopt}. A simple argument for convergence
on bipartite graph is provided in Section \ref{sec:Bipartite}, while the
much more challenging case of general graphs is contained in
Section \ref{sec:BasicStructures} and the appendices.

\subsubsection{The KT construction and the gap of a solution}
\label{sec:KT+gap}

Let $G$ be an instance which admits at least one stable outcome,
$M^*$ be the corresponding matching (recall that this is a maximum weight
matching), and consider the Kleinberg-Tardos (KT)
procedure for finding a NB solution \cite{KT}.
Any NB solution $\bargamma^*$ can be constructed by this procedure with
appropriate choices at successive stages. At each stage, a linear program
is solved with variables $\gamma_i$ attached to node $i$.
The linear program
maximizes the minimum `slack' of all unmatched edges and nodes,
whose values have not yet been set
(the slack of edge $(i,j)\not\in M$ is $\gamma_i+\gamma_j-w_{ij}$).

At the first stage, the set of nodes that
remain unmatched (i.e. are not part of $M^*$) is
found, if such nodes exist.
Call the set of unmatched nodes $\cC_0$.
\begin{figure}
\hspace{2cm}\includegraphics[scale=0.4]{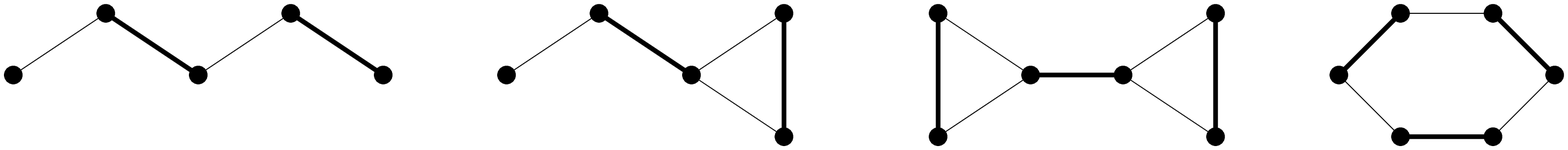}
\caption{Examples of basic structures: path, blossom, bicycle, and cycle
(matched edges in bold).}
\label{fig:structures}
\end{figure}
After this, at successive stages of the KT procedure,
a sequence of structures $\cC_1$, $\cC_2,\ldots, \cC_k$ characterizing the LP optimum are found.
We call this the \textit{KT sequence}.
Each such structure is a pair $\cC_q = (V(\cC_q),E(\cC_q))$
with $V(\cC_q)\subseteq V$, $E(\cC_q)\subseteq E$.
According to \cite{KT} $\cC_q$ belongs to one of four topologies:
alternating path, blossom, bicycle, alternating cycle (Figure \ref{fig:structures}).
The $q$-th linear program determines the value of $\gamma_i^*$
for $i\in V(\cC_q)$.
Further, one has the partition $E(\cC_q)=E_1(\cC_q)\cup E_2(\cC_q)$
with  $E_1(\cC_q)$ consisting of all matching edges along which
nodes in $V(\cC_q)$ trade, and
$E_2(\cC_q)$ consists of edges $(i,j)$ such that some $i\in V(\cC_q)$
receives its second-best, positive offer from $j$.

The $\gamma$ values for nodes on the limiting structure are uniquely
determined if the structure is an alternating path, blossom or
bicycle\footnote{In \cite{KT}
it is claimed that the $\gamma$ values `may not be fully determined'
also in the case of bicycles. However it is not hard to prove that
this is not the case.}.
In case of an
alternating cycle there is one degree of freedom -- setting
a value $\gamma_i^*$ for one node $i \in \cC_q$ fully
determines the values at the other nodes.

We emphasize that,
within the present definition,
$\cC_q$ is not necessarily a subgraph of $G$, in that it might contain
an edge $(i,j)$ but not both its endpoints.
On the other hand, $V(\cC_q)$ is always subset of the endpoints
of $E(\cC_q)$.
We denote by $V_{\textup{ext}}(\cC_q)\supseteq V(\cC_q)$
the set of nodes formed by all the endpoints of edges in
$E(\cC_q)$.

For all nodes $i\in V(\cC_q)$ the second best offer
is equal to $\gamma_i^* -\sigma_q$, where $\sigma_q$ is the slack of
$\cC_q$. Therefore
\begin{eqnarray}
\gamma_{i}^*+\gamma_j^*-w_{ij} = \left\{
\begin{array}{ll}
0&\mbox{ if $(i,j)\in E_1(\cC_q)$,}\\
\sigma_q &\mbox{ if $(i,j)\in E_2(\cC_q)$}\\
\end{array}\right.
\label{eq:edge_slack}
\end{eqnarray}
%

The slacks form an increasing sequence ($\sigma_1\le \sigma_2 \le \ldots \le \sigma_k$).
We say that $\baralf^*$ has a gap $\sigma$ if
\begin{eqnarray}
\sigma \le \min\big\{\sigma_1;\, \sigma_2-\sigma_1;\, \dots;\,\sigma_k-
\sigma_{k-1}\big\}\, ,
\label{eq:sigma_cond_1}
\end{eqnarray}
and if for each edge $(i,j)$ such that $i,j\in V_\textup{ext}(\cC_q)$
and $(i,j)\not\in E(\cC_q)$,
\begin{eqnarray}
\gamma_i^*+\gamma_j^* - w_{ij}\ge \sigma_q+\sigma\, .
\label{eq:extra_sigma_condition}
\end{eqnarray}
It is possible to prove that the positive gap condition is
generic in the following sense. The set of all
instances such that the NB solution
is unique can be regarded as a
subset ${\sf G}\subseteq[0,W]^{|E|}$ ($W$ being the maximum edge
weight). It turns out that ${\sf G}$ has dimension $|E|$ (i.e.
the class of instances having unique NB solution is large)
and that the subset of instances with gap $\sigma>0$ is both
\emph{open and dense in ${\sf G}$}.

%
%
 \subsubsection{Convergence}
\label{sub:Convergence}

We begin with a basic property.
Throughout the paper, for a vector $\ux$, we let
$||\ux||_{\infty}\equiv\max_i|x_i|$.
\begin{lemma}
The natural dynamics never leads to an expansion in the sup-norm.
More precisely, for any two initial message vectors $\baralf^0$ and
$\barbet^0$, we have
\begin{align}
||\baralf^1 -\barbet^1||_\infty \leq ||\baralf^0 -\barbet^0||_\infty\, .
\end{align}
\label{lem:infty_non_expansion}
\end{lemma}
The proof of this fact is elementary and deferred to Appendix
\ref{sec:BasicConvergence}.
In view of this lemma, it is natural to consider the unique
(by assumption) fixed point $\baralf^*$ and consider the \emph{cost function} $U_G(\baralf) = ||\baralf -\baralf^*||_{\infty}$.

In order to prove Theorem \ref{thm:Convergence}, it is sufficient to
show that this cost decreases by a non-negligible amount in a polynomial
number of iterations. In particular, it is sufficient to prove the following.
\begin{theorem}\label{thm:TechnicalConvergence}
Let $G$ be an instance having unique solution $\baralf^*$ with gap $\sigma$
and consider an initial condition $\baralf^0$ with $U_G(\baralf^0)\le \Delta$.
Then there exists $C<\infty$ such that for all $t \geq C\, n^{7}$,
\begin{eqnarray}
U_G(\baralf^t)\le \Delta-\min(\sigma,\Delta)/2\, .
\label{eq:after_descent}
\end{eqnarray}
\end{theorem}
Theorem \ref{thm:TechnicalConvergence} is proved by analyzing sequentially
each separate structure in the KT sequence. The key remark is that
such a separate analysis is possible because the convergence
of each structure \emph{decouples} from all subsequent structures
on a scale smaller than $\sigma$.

Given a subgraph $F=(V_F,E_F)\subseteq G$, we let
\begin{eqnarray}
U_{G,F}(\baralf) = \max_{(i,j)\in E_F}|\alf{i}{j} -\alf{i}{j}^*|\, .
\end{eqnarray}
Note that $U_{G,F}(\baralf)$ depends $F$ only through its edge set
$E_F$. Thus, there will be no ambiguity in the notation
$U_{G,\cC_q}(\baralf)$.
We let $\cG$ be the directed multi-graph with the same vertex set as $G$
and all edges obtained by directing the edges in $G$
(hence for each $(i,j)$ in $G$ we have two directed edges $i\to j$
and $j\to i$ in $\cG$).
With an abuse of notation, we shall write $U_{\cG,\cF}(\baralf)$
also when $\cF$ is a directed subgraph of $\cG$.
We denote by $\cG_q, q \in \{1,2, \ldots, k+1\}$ the  directed graph
including all directed edges between vertices in
$V(\cC_0)\cup V(\cC_1)\cup\dots\cup V(\cC_{q-1})$
and all directed edges with the first endpoint in this set of vertices.
We also let $\cG_0$ be the empty graph.
\begin{definition}
For $\Delta > 0$, $\delta \leq \min(\Delta, \sigma)$, we define
$\Bo_\Delta(q, \delta,t)$ to be the condition
\begin{eqnarray}
U_{\cG,\cG_q}(\baralf^t)\le \Delta-\delta\,,\;\;\;\;\;\;\;\;
U_{\cG}(\baralf^t)\le \Delta\, .\label{eq:ConditionStructures}
\end{eqnarray}
In the following, we will sometimes drop the subscript $\Delta$.
\end{definition}
We can now state the key lemma for analyzing the convergence of
structures in the KT sequence.
\begin{lemma}\label{lemma:KeyLemma}
Let $G$ be an instance having unique solution with gap $\sigma$.
Consider an initial
condition $\baralf^0$ such that $\Bo_\Delta(q,\delta,0)$ holds for some
$\Delta > 0$, $\delta \le \min(\Delta, \sigma)$
and for some $q \in \{0,1, \ldots, k\}$.
Then there exists $t_*(n) = C\,n^6$ such that $\Bo_\Delta(q+1,\delta
(1-(5n)^{-1}),t)$ holds  for all $t\ge t_*(n)$.
\end{lemma}
The proof of this lemma is based on a case-by-case analysis of the
possible topologies of $\cC_q$ and can be found in Sections
\ref{sec:BasicStructures},
\ref{app:Path}, \ref{app:Blossom}, \ref{app:Bicycle}.
Using the lemma to prove Theorem \ref{thm:TechnicalConvergence},
and hence Theorem \ref{thm:Convergence} is immediate.
\begin{proof}[Proof (Theorem \ref{thm:TechnicalConvergence})]
We know that  $k \leq n$. Start with
 $\baralf^0$ such that
$U_{\cG,\cG_q}(\baralf^0)\le \Delta$ at $t=0$, i.e. $\Bo_{\Delta}(q,\delta,0)$.
Define $\delta_0=\min(\sigma, \Delta)$,
$\delta_1=\delta_0(1-(5n)^{-1})$,
\ldots, $\delta_k=\delta_0(1-(5n)^{-1})^k$.

We know that $\Bo(0,\delta_0,0)$ holds by assumption.
We deduce from Lemma \ref{lemma:KeyLemma}
that $\Bo(1,\delta_1,t_*)$ holds. Proceeding inductively, it follows that
$\Bo(k+1,\delta_k,(k+1)t_*)$ holds. Now, we only need to show
$\delta_k \geq \frac{1}{2}\delta_0$, which follows from
$(1-(5n)^{-1})^{n} \geq 1/2$ for all $n\geq 1$.
\end{proof}

%
%
\section{Fixed point properties: Proof of Theorem \ref{thm:fp_dualopt}}
\label{sec:FixedPoint}

The dual problem to (\ref{prob:mwm_relaxation}) is
\begin{eqnarray}
\textup{minimize} && \sum_{i \in V} y_i,\nonumber\\
\textup{subject to}&&
y_i+y_j \ge w_{ij} \;\;\; \forall (i,j) \in E,\;\;\;\;\;\;\;
y_i \ge 0 \;\;\;\forall i \in V
\label{prob:mwm_dual}
\end{eqnarray}
A feasible point $\underline{x}$ for LP \eqref{prob:mwm_relaxation} is called \emph{half-integral} if
for all $e\in E$, $x_e\in\{0,1,\frac{1}{2}\}$. It is well known that problem (\ref{prob:mwm_relaxation}) always has an optimum $\underline{x}^*$ that is half-integral. We also denote the subset of edges $e$ in $E$ with $x_e^* \in \{ \frac{1}{2}, 1 \}$ by $M^*$ and say $M^*$ is a \emph{half-integral matching}.

In this paper we consider only those problems that have unique $\underline{x}^*$, that has greater weight than any other corner $\underline{x}$ of the primal polytope by at least $\epsilon>0$, i.e. $\sum_{e \in E} w_e x_e^*-\sum_{e \in E} w_e x_e \geq \epsilon$ for all half-integral solutions $\underline{x} \neq \underline{x}^*$.
We call such LP an \emph{$\epsilon$-pointed problem}. Further, if $\underline{x}^*$ is, in fact, integral, we say that the LP relaxation is \emph{$\epsilon$-tight}. We will simply use \emph{pointed}
and \emph{tight}, whenever the value of $\eps>0$ is immaterial
(let us stress that, according to this terminology, tight implies pointed).

We call  $e\in E$ a \emph{$1$-solid edge} if $x_e^*=1$, a \emph{$\frac{1}{2}$-solid} edge if $x_e^*=\frac{1}{2}$,
and a \emph{non-solid edge} if $x_e^*=0$.
%
\paragraph{Proof of Theorem \ref{thm:fp_dualopt}: From fixed points to NB solutions.}
The direct part follows from the following set
of fixed point properties, which hold for pointed problems.
The proofs of these properties are given in Appendix \ref{sec:fixed_point_prop_proofs}. Throughout  $(\baralf,\baroff,\bargamma)$ is a fixed point of the
dynamics (\ref{eq:off_def}), (\ref{eq:update}) (with $\bargamma$
given by (\ref{eq:bargamma})).
\begin{enumerate}
\item Two players $(i,j) \in E$ are called \emph{partners} if $\gamma_i + \gamma_j = w_{ij}$. Then the following are equivalent:
(a) $i$ and $j$ are partners, (b) $w_{ij}-\alf{i}{j} - \alf{j}{i}\geq 0$, (c) $\gamma_i=\off{j}{i}$ and $\gamma_j=\off{i}{j}$.

\item Let $P(i)$ be the set of all partners of $i$. Then the following are equivalent: (a) $P(i) = \{ j \}$, (b) $P(j) = \{ i \}$, (c) $w_{ij}-\alf{i}{j} - \alf{j}{i} >0$.

\item We say that $(i,j)$ is a \emph{weak-dotted edge} if $w_{ij}-\alf{i}{j} - \alf{j}{i}=0$, \emph{a strong-dotted edge} if
$w_{ij}-\alf{i}{j} - \alf{j}{i} > 0$, and a \emph{non-dotted edge} otherwise. If $i$ has no adjacent dotted edges, then $\gamma_i=0$.

\item $\bargamma$ is an optimum solution for the dual LP (\ref{prob:mwm_dual}) and $\off{i}{j}=(w_{ij}-\gamma_i)_+$ holds for all $(i,j) \in E$.

\item The balance property \eqref{eq:balance}, holds at every edge $(i,j) \in E$ (with both sides being non-negative).

\item An edge is $1$-solid ($\frac{1}{2}$-solid) iff it is strongly (weakly) dotted.
\end{enumerate}

\begin{proof}[Proof of Theorem \ref{thm:fp_dualopt}, direct implication]
Assume that the LP (\ref{prob:mwm_relaxation}) is tight. Then, by property 6,
the set of strong-dotted edges form the unique maximum weight matching $M^*$.
By properties 4 and 5 the pair $(M^*,\gamma)$ is stable and balanced
and thus forms a NB solution.
\end{proof}

\noindent\emph{Remark :} Notice that properties 1-6 hold under the weaker condition
that the problem \eqref{prob:mwm_relaxation} is pointed. In this more general setting, fixed
point correspond to dual optima satisfying the unmatched
balance property (\ref{eq:balance}).
%
%
\paragraph{Proof of Theorem \ref{thm:fp_dualopt}: From NB solutions to fixed points.}
\begin{proof}[Proof of Theorem \ref{thm:fp_dualopt}, converse implication]
Consider any NB solution $(M,\bargamma_{\NB} )$. Construct a corresponding
FP as follows. Set $\off{i}{j}=(w_{ij}-\gamma_{\NB,i})_+$ for all $(i,j) \in E$.
Compute $\baralf$ using $\alf{i}{j}=\max_{k \in \di \backslash j} \off{k}{i}$.
We claim that this is a FP and that the corresponding $\bargamma$ is $\bargamma_{\NB}$.
To prove that we are at a fixed point, we imagine updated offers $\baroff^{\textup{upd}}$
based on $\baralf$, and show $\baroff^{\textup{upd}}=\baroff$.

Consider a matching edge $(i,j) \in M$. We know that $\gamma_{\NB,i} + \gamma_{\NB,j}=w_{ij}$.
Also stability and balance tell us $\gamma_{\NB,i} - \max_{k \in \di \backslash j} (w_{ik}- \gamma_{\NB,k})_+ =
\gamma_{\NB,j} - \max_{l \in \dj \backslash i} (w_{jl}- \gamma_{\NB,l})_+$ and both sides are non-negative.
Hence, $\gamma_{\NB,i} - \alf{i}{j} = \gamma_{\NB,j} - \alf{j}{i} \geq 0$.
Therefore $\alf{i}{j}+\alf{j}{i} \leq w_{ij}$,
\begin{align}
\off{i}{j}^\textup{upd}=\left( \frac{w_{ij}-\alf{i}{j}+\alf{j}{i}}{2} \right)
= \left( \frac{w_{ij}-\gamma_{\NB,i}+\gamma_{\NB,j}}{2} \right)= \gamma_{\NB,j} = w_{ij}-\gamma_{\NB,i} = \off{i}{j}.
\end{align}
By symmetry, we also have $\off{j}{i}^\textup{upd}= \gamma_{\NB,j} = \off{j}{i}$. Hence, the offers remain unchanged.
Moreover, $\gamma_i = \max(\off{j}{i}, \alf{i}{j}) =\gamma_{\NB,i}$. Similarly, $\gamma_j= \gamma_{\NB,j}$.  Now consider $(i,j) \notin M$. We have $\gamma_{\NB,i}+\gamma_{\NB,j} \geq w_{ij}$ and, $\gamma_{\NB,i} = \max_{k \in \di \backslash j} (w_{ik}- \gamma_{\NB,k})_+ = \alf{i}{j}$. Similar equation holds for $\gamma_{\NB,j}$. The validity of this identity can be checked individually in the cases when $i \in M$
and $i \notin M$. Hence, $\alf{i}{j}+\alf{j}{i} \leq w_{ij}$.
This leads to $\off{i}{j}^{\textup{upd}}=(w_{ij}-\alf{i}{j})_+=(w_{ij}-\gamma_{\NB,i})_+=\off{i}{j}$. By symmetry, we know that both the offers remain unchanged.

Finally, we show $\bargamma = \bargamma_{\NB}$. For all $(i,j) \in M$, we already found that $\off{i}{j}=\gamma_j$ and vice versa.
For any edge $(ij) \notin M$, we know $\off{i}{j} = (w_{ij}-\gamma_{\NB,i})_+ \leq \gamma_{\NB,j}$. This immediately leads
to $\bargamma = \bargamma_{\NB}$. It is worth noting that making use of the uniqueness of LP optimum, we can further show that
 $M=M^*$, $\gamma_i = \off{j}{i} > 0$ iff $(ij) \in M$ i.e. the fixed point reconstructs the pairing $M=M^*$.
\end{proof}
%
%
\section{An easier case: bipartite graphs}
\label{sec:Bipartite}

Bipartite graphs are natural graphs to consider, since they
include the case of `buyers' and `sellers' negotiating with each other.
This case is substantially easier because of a  natural partial ordering
on message vectors, and in particular,
on the Nash bargaining solutions. This ordering loosely corresponds to the
perceived strength or negotiating power of buyers relative to sellers.
Further, there are `extremal' NB solutions with respect to
this partial ordering,
which we call
the buyer/seller dominant solution. We show that if the dynamics starts with a sufficiently
buyer/seller dominant state, then it always converges to one of these extreme NB solutions.

Unfortunately, the partial ordering does not extend to
general graphs, and a much more technical analysis is required in that case.
Nevertheless, the bipartite graph case suggests that the natural dynamics
indeed converges even in instances with multiple NB solutions.

We partition the nodes on a bipartite graph $G$ as $V=(V_{\rm B}, V_{\rm S})$,
representing buyers and seller respectively.
Given message vectors $\baralf$ and $\barbet$, we define a partial ordering
as $\baralf \succeq \barbet$ if  inequalities $\alf{i}{j'} \geq \bet{i}{j'}$ and $\alf{j'}{i} \leq \bet{j'}{i}$ 
hold for all $ (i,j') \in E$, $i \in V_{\rm B}, j' \in V_{\rm S}$.
\begin{lemma}
For a bipartite graph $G$, if $\baralf^0 \succeq \barbet^0$ then
$\baralf^{t} \succeq \barbet^{t}$ for all $t \geq 0$
i.e. the dynamics preserves partial ordering.
\label{lemma:bipartite_monotonicity}
\end{lemma}
\begin{proof}
$\off{i}{j'}$ is non-increasing in $\alf{i}{j'}$ and non-decreasing in
$\alf{j'}{i}$ leading to $\off{i}{j'}|_{\baralf^0} \leq \off{i}{j'}|_{\barbet^0}$.
Similarly, $\off{j'}{i}|_{\baralf^0} \geq \off{j'}{i}|_{\barbet^0}$. It
follows that $\baralf^1 \succeq \barbet^1$. Induction completes the proof.
\end{proof}

\begin{corollary}\label{corr:monotone_convergence}
Assume $\baralf^0$ to be such that $\baralf^0 \succeq \baralf^1$. Then
$\lim_{t \rightarrow \infty} \baralf^t = \baralf^*$
for some NB solution $\baralf^*$.
\end{corollary}
\begin{proof}
Since $\baralf^0 \succeq \baralf^1$, using Lemma \ref{lemma:bipartite_monotonicity}
we have $\baralf^t \succeq \baralf^{t+1}$ for all $t$.
We know that $\baralf \in [0,W]^{2|E|}$ for all valid message
vectors $\baralf$, so the monotone sequence $\{\baralf^t\}_{t\geq 0}$ converges
to a fixed point of the dynamics (as the update rules are continuous in
$\baralf^t$).
\end{proof}
Define the message vectors $\baralf^\top$ and $\baralf^\bot$ by: $\alf{i}{j'}^{\top}=W \;\;, \alf{j'}{i}^{\top}=0$ and
$\alf{i}{j'}^{\bot}=0\;\;,\alf{j'}{i}^{\bot}=W$ for all $(i,j') \in E$. Therefore, $\baralf^\top \succeq \baralf \succeq \baralf^\bot$ for any valid $\baralf \in [0,W]^{2|E|}$.

If $\baralf^0=\baralf^{\top}$ then $\baralf^0 \succeq \baralf^1$.
It follows from Corollary \ref{corr:monotone_convergence} that $\baralf^t$ converges, say to
$\baralf^{*,\up}$. Consider any other fixed point (i.e. NB solution) $\baralf^*$.
$\baralf^0 \succeq \baralf^* \Rightarrow \baralf^t \succeq \baralf^* \; \forall \, t \geq 0$.
Hence, $\baralf^{*,\up}\succeq \baralf^*$. Thus, there exists
a `maximal' fixed point
with respect to the partial ordering.
It is buyer optimal in a strong sense -- Each buyer
earns the most possible at $\baralf^{*,\up}$ among all NB solutions, and each
seller earns the least possible.
With a similar argument using $\alf{i}{j'}^{\bot}$, we can show the existence of a
`minimal' fixed point $\baralf^{*,\bot}$ that is seller optimal.

Although the simple monotonicity argument does not imply any 
bound on the convergence rate, we expect that the methods developed for 
general graphs (next Section) can be used to prove polynomial convergence 
in the present case.
%
%
\section{General convergence: Proof of Lemma \ref{lemma:KeyLemma}}
\label{sec:BasicStructures}

The condition $\Bo(q,\delta,t)$ from Section \ref{sec:KT+gap} means that at a fixed step $t$ of the natural dynamics the deviation of $\baralf^t$ from the unique fixed point $\baralf^*$ is at most $\Delta$ on all edges of $E$, and is at most $\Delta-\delta$ on all edges of $\cG_q$, or $U_{\cG}(\baralf^t)\le \Delta$,
$U_{\cG,\cG_q}(\baralf^t)\le \Delta-\delta$.

The goal now is to show that if we start with an initial condition that satisfies $\Bo(q,\delta,0)$ then after $Cn^6$ iterations of the natural dynamics we reach the condition  $\Bo(q+1,\delta(1-(5n)^{-1}),t)$.

First we state the following useful properties of the condition $\Bo(q,\delta,t)$ with respect to the natural dynamics. Their proofs are given in \ref{sec:BasicConvergence}. The first lemma gives an upper bound on offers to nodes in $V(\cC_q)$ from outside $\cC_q$. This allows us to restrict the analysis to messages on $\cC_q$.
\begin{lemma}
Let $q\in\{0,1, \ldots, k\}$ and $\delta \leq \min(\sigma, \Delta)$.
If $\Bo(q,\delta,t)$ holds then for all $(i,j) \in E$ with $i$ in $V(\cC_q)$ and $(i,j)\notin E(\cC_q)$:
$\off{j}{i}^t \leq \gamma_i^* -\sigma_q+\Delta -\delta$.
\label{lemma:outsideOff_ub}
\end{lemma}

\begin{lemma}
If $\Bo(q,\delta,0)$ holds, then  $\Bo(q,\delta,t)$ holds for all  $t >0$.
\label{lemma:CondStrucPersistence}
\end{lemma}
\begin{lemma}
Assume that $\Bo(q,\delta, 0)$ holds for some $\delta \leq \min(\sigma, \Delta)$, $0 \leq q \leq k$, and that $U_{G,\cC_q}(\baralf^t) \leq \Delta - \delta$ for all $t \geq 0$.  Then, for any $\epsilon > 0$, $\Bo(q+1, \delta - \epsilon, t)$ holds
for all $t\geq C \log(1+\delta/\epsilon)$.
\label{lemma:outgoing_error_bound}
\end{lemma}
Lemma \ref{lemma:KeyLemma} for $q=0$ follows directly from Lemma \ref{lemma:outgoing_error_bound}
with $\epsilon= \delta/(5n)$. For $q>0$, the Lemma \ref{lemma:outgoing_error_bound} reduces our problem to bounding $U_{G,\cC_q}(\baralf^t)$.
The proof is divided into three cases depending
on the structure $\cC_q$: path, blossom, bicycle.  Next we
discuss the main steps in the case of a path.
The blossom and bicycle require a more delicate analysis that uses the same ideas. These cases are deferred to
Appendix sections \ref{app:Blossom}, \ref{app:Bicycle} in the interest of space.

\subsection{Proof of Lemma \ref{lemma:KeyLemma}: Path}
\label{subsubsec:path_descent}

We prove the bound $U_{G,\cC_q}(\baralf^t) \leq \Delta - \delta(1-(10n)^{-1})$ in three steps:

\textbf{Step 1:} Define two different dynamics on the path that are simpler to analyze, we call these `bounding processes'.
Each bounding process has a unique fixed point and converges rapidly to it.

\textbf{Step 2:} Show that the message vector $\baralf^t$ of the natural dynamics on this path is bounded from above and below (sandwiched) by these two bounding processes.

\textbf{Step 3:} Show that the fixed points of the bounding processes have deviations $\Delta-\delta$ from $\baralf^*$.
Hence, the natural dynamics quickly achieves $U_{G,\cC_q}(\baralf^t) \leq \Delta - \delta(1-(10n)^{-1})$.

\paragraph{Bounding processes.}
Define a path as $P = (V_P,E_P)$ with $V_{P} = \{0,1,\dots,\len\}$, and $E_P = \{(i,i+1):\;i =0,\dots,\len-1\}$
with positive weights on the edges, and $\len \geq 1$.
Let $M_P$ be a matching on $P$ with alternate edges included in $M_P$.
We define the \emph{simplified dynamics} on $P$ wrt matching $M_P$, boundary conditions $\{\boundary_\L^t\}_{t\geq 0}, \{\boundary_\R^t\}_{t\geq 0}$ (arbitrary real numbers) by
\begin{eqnarray}
\forall \,(i,j) \in E_P,  \ \ \soff{i}{j}^{t} =
\left\{\begin{array}{ll}
\frac{1}{2}(w_{i,j}-\salf{i}{j}^t+\salf{j}{i}^t)&\mbox{ if $(i,j)\in M_P$,}\\
w_{i,j}-\salf{i}{j}^t& \mbox{ otherwise.}\\
\end{array}\right. \label{eq:simplified_offers}
\end{eqnarray}

\begin{align*}
\salf{0}{1}^{t+1}&=\damp\boundary_\L^t+(1-\damp)\salf{0}{1}^t\, ,& \salf{\len}{\len-1}^t&=\damp\boundary_\R^t+(1-\damp)\salf{\len}{\len-1}^t\\
\salf{i}{i+1}^{t+1}& = \damp \soff{i-1}{i}^t+(1-\damp)\salf{i}{i+1}^{t}\, ,&
\salf{i}{i-1}^{t+1}& = \damp \soff{i+1}{i}^t+(1-\damp)\salf{i}{i-1}^{t}\, ,~~~~ i = 1, 2, \ldots, \len -1
\end{align*}
for all $t\ge 0$.

Our bounding processes will be defined later as instances of this simplified dynamics with specially tuned boundary and initial conditions.  But first we state a crucial result on the convergence of the simplified dynamics. Its proof is given in \ref{app:Path} based on the analysis of a  random walk on a finite length segment.
\begin{lemma}\label{lemma:simplified_convergence}
Consider $P$ with $\len \geq 1$. Let $\boundary_\L^t=\boundary_\L, \boundary_\R^t=\boundary_\R, t\geq 0$ be arbitrary constant boundary conditions.
The simplified dynamics on $P$ wrt the given boundary conditions and matching $M_P$ has a unique fixed point $\sbaralf^*$ satisfying
$\salf{0}{1}^*=\boundary_\L, \salf{\len}{\len-1}^*=\boundary_\R$. Moreover, for any initial condition
$\sbaralf^0$,
\begin{eqnarray}
||\sbaralf^t-\sbaralf^*||_{\infty}\le
C \len^2 \exp\bigg(-\frac{t}{C \len^2}\bigg)||\sbaralf^0-\sbaralf^*||_{\infty} \, ,~~~ \mbox{ $C<\infty$}.
\end{eqnarray}
\end{lemma}
As a consequence, the error term will be smaller than any desired inverse polynomial factor in time $O(n^2 \log n)$.

Now we define the bounding processes. Suppose $\cC_q$ is a path $P$. Let $M_P^*$ be the restriction of $M^*$ to $P$. Assume $\Delta>0$, $\delta\le \min(\sigma,\Delta)$ and $s\in \{+1,-1\}$.
The $(s,\Delta,\delta)$-bounding process on $P$ wrt $M_P^*$ is
the sequence of message assignments $\{\barbound^t\}_{t\ge 0}$ on $P$
produced by the simplified dynamics on $P$ wrt the matching $M_P^*$,
the constant boundary conditions $b_\L= s\, (\Delta-\delta), b_\R= s(-1)^{\len}\,
(\Delta -\delta)$ with the initial condition
(for $i\in\{0,\dots,\len-1\}$) $\bound{i}{i+1}^0 = \alf{i}{i+1}^*+s(-1)^i\Delta$, $\bound{i+1}{i}^0 = \alf{i+1}{i}^*-s(-1)^i\Delta$. When convenient,
we include the sign $s$ explicitly in the notation for the
bounding process as $\{\barbound^t(s)\}_{t\ge 0}$.

As per Lemma \ref{lemma:simplified_convergence}, the bounding process has a unique fixed point to which it converges fast. Call this fixed point $\barbound^*$. It is easy to see that for all $i \in \{0,1,\ldots,\len-1\}$: $~\bound{i}{i+1}^* = \alf{i}{i+1}^*+s(-1)^i(\Delta-\delta)$, $\bound{i+1}{i}^* = \alf{i+1}{i}^*-s(-1)^{i}(\Delta -\delta)$.

The following lemma establishes that the simplified dynamics and natural dynamics produce the same results when acting on $\barbound^t$, except when $\soff{i}{j}^t<0$ for some $(i,j)\in P$. This property is critical in proving that $\barbound^t$ sandwiches the natural dynamics (Lemma \ref{lemma:BoundingBounds}). The proof is in Appendix \ref{app:Path}.
\begin{lemma}\label{lemma:CorrectMatching}
Let $\{\barbound^t\}_{t\ge 0}$ be the $(s,\Delta,\delta)$-bounding
process on $P$. Then for any $t\geq 0$, we have
$\bound{i}{j}^t +\bound{j}{i}^t-w_{ij}\le 0$ for $(i,j)\in M_P^*$ and $\bound{i}{j}^t +\bound{j}{i}^t-w_{ij}\ge 0$ for $(i,j)\not\in M_P^*$.
\end{lemma}

\paragraph{Sandwiching the message vector.}
Given two message assignments $\baralf$ and $\barbet$ on
$P$, we shall say that $\baralf$ dominates $\barbet$
(and write $\baralf\succeq \barbet$) if
\begin{eqnarray}
\mbox{ for $i$ even }&~~ \alf{i}{i+1} \ge \bet{i}{i+1},&\alf{i+1}{i} \le \bet{i+1}{i}\label{def:ordering_even}\\
\mbox{ for $i$ odd }&~~ \alf{i}{i+1} \le \bet{i}{i+1},&\alf{i+1}{i} \ge \bet{i+1}{i}\label{def:ordering_odd}
\end{eqnarray}
The natural dynamics (and the simplified dynamics) on $P$ preserve this ordering (cf. Appendix \ref{app:Path}).

Note that $\barbound^0(-)\preceq \baralf^0_P\preceq \barbound^0(+)\,$ by our definition of the bounding processes, where
$\baralf_P$ denotes the restriction of $\baralf$ to $P$. The next lemma shows that this sandwich property continues to hold for all $t\geq 0$. The proof is given in Appendix \ref{app:Path}.
\begin{lemma}\label{lemma:BoundingBounds}
Let $G$ be an instance admitting a unique NB solution with
gap $\sigma$, and assume its KT sequence
to be $(\cC_0, \cC_1,\cC_2,\dots,\cC_k)$, with $\cC_q, q \in \{1,\ldots, k\}$ a path. Let
$P=(V_\textup{ext}(\cC_q), E(\cC_q))$.
Further assume $\Bo_\Delta(q,\delta,0)$ holds.
If we denote by $\{\barbound^t(s)\}_{t\ge 0}$
the $(s,\Delta,\delta)$ bounding process on $P$, then for any
$t\ge 0$:
\begin{eqnarray}
\barbound^t(-)\preceq \baralf^t_P\preceq \barbound^t(+)\, .
\label{eq:real_path_bound}
\end{eqnarray}
\end{lemma}

\begin{proof}[Proof (Lemma \ref{lemma:KeyLemma}: Path)]
From Lemma \ref{lemma:BoundingBounds} we have
\begin{eqnarray}
||\baralf^t-\baralf^*||_{\infty}
\le
\max\Big\{||\barbound^t(+)-\baralf^*||_{\infty},\;
||\barbound^t(-)-\baralf^*||_{\infty}\Big\}
\ \, .
\end{eqnarray}
We know that $\len\leq n$.
Since $||\barbound^*-\baralf^*||_{\infty} = \Delta-\delta$ by Lemma \ref{lemma:simplified_convergence} applied to the bounding processes we can show the right hand side becomes smaller than $\Delta-\left(1-\frac{1}{10n}\right)\delta$ for all $t\geq c_\P n^3 $ for some $c_\P$ finite.
Finally, we use Lemma \ref{lemma:outgoing_error_bound} with
$\eps = \frac{\delta}{10n}$ to show that $\Bo(q+1, \delta(1-(5n)^{-1}), t)$ holds
for all $t \geq t_{\P,*}$, where $t_{\P,*} \leq C n^6$ as required.
\end{proof}

\vspace{0.05cm}

{\bf\large Acknowledgements.}
We thank Eva Tardos for introducing us to network exchange theory
and Daron Acemoglu for several insightful discussions. A large part of this
work was done while YK, MB and AM were at MSR New England.
This research was partially supported by NSD, grants CCF-0743978 
and CCF-0915145, by a Terman fellowship. YK is supported by a 3Com
Corporation Stanford Graduate Fellowship. 

\newpage

\bibliographystyle{amsalpha}

\newpage

\appendix

\section{Proofs of fixed point properties}
\label{sec:fixed_point_prop_proofs}
In this section we state and prove the fixed point properties that were used for the proof of Theorem \ref{thm:fp_dualopt} in Section \ref{sec:FixedPoint}.

\begin{lemma}
$\bargamma$ satisfies the constraints of the dual problem (\ref{prob:mwm_dual}).
\label{lemma:gamma_sat_dual_const}
\end{lemma}
\begin{proof}
Since offers $\off{i}{j}$ are by definition non-negative therefore for all $v\in V$ we have $\gamma_v\geq 0$. So we only need to show $\gamma_i+\gamma_j\geq w_{ij}$ for any edge $(ij)\in E$.  It is to see that $\gamma_i\geq \alf{i}{j}$ and $\gamma_j\geq \alf{i}{j}$. Therefore, if $\alf{i}{j}+\alf{i}{j}\geq w_{ij}$ then $\gamma_i+\gamma_j\geq w_{ij}$ holds and we are done.  Otherwise, for $\alf{i}{j}+\alf{i}{j}<w_{ij}$ we have
$\off{i}{j} = \frac{w_{ij}-\alf{i}{j}+\alf{j}{i}}{2}$ and $\off{j}{i} = \frac{w_{ij}-\alf{j}{i}+\alf{i}{j}}{2}$ which gives
$\gamma_i+\gamma_j\geq \off{i}{j}+\off{j}{i}=w_{ij}$.\enp
\end{proof}

Recall that for any $(ij) \in E$, we say that $i$ and $j$ are `partners' if $\gamma_i + \gamma_j = w_{ij}$ and
$P(i)$ denotes the partners of node $i$. In other words $P(i) = \{ j: j \in \partial i, \gamma_i + \gamma_j = w_{ij}\}$.

\begin{lemma}\label{lemma:i-jpartners_equivalence}
The following are equivalent:
\begin{itemize}
\item[(a)] $i$ and $j$ are partners,
\item[(b)] $w_{ij}-\alf{i}{j} - \alf{j}{i}\geq 0$,
\item[(c)] $\gamma_i=\off{j}{i}$ and $\gamma_j=\off{i}{j}$.
\end{itemize}
Moreover, if $\gamma_i=\off{j}{i}$ and $\gamma_j>\off{i}{j}$ then $\gamma_i=0$.
\end{lemma}
\begin{proof} We will prove $(a) \Rightarrow (b) \Rightarrow (c) \Rightarrow (a)$.

$(a)\Rightarrow (b)$: Since $\gamma_i\geq \alf{i}{j}$ and $\gamma_j\geq \alf{j}{i}$ always holds then $w_{ij}=\gamma_i + \gamma_j\geq \alf{i}{j} + \alf{j}{i}$.

$(b)\Rightarrow (c)$: If $w_{ij}- \alf{i}{j} - \alf{j}{i}\geq 0$ then $(w_{ij}-\alf{i}{j}+\alf{j}{i})/2\geq \alf{j}{i}$. But $\off{i}{j}=(w_{ij}-\alf{i}{j}+\alf{j}{i})/2$ therefore $\gamma_j=\off{i}{j}$. The argument for $\gamma_i=\off{j}{i}$ is similar.

$(c)\Rightarrow (a)$: If $w_{ij} \geq \alf{i}{j} + \alf{j}{i}$ then $\off{i}{j}=(w_{ij}-\alf{i}{j}+\alf{j}{i})/2$ and $\off{j}{i}=(w_{ij}-\alf{j}{i}+\alf{i}{j})/2$ which gives $\gamma_i+\gamma_j=\off{i}{j}+\off{j}{i}=w_{ij}$ and we are done.  Otherwise, we have
$\gamma_i+\gamma_j=\off{i}{j}+\off{j}{i}\leq (w_{ij}-\alf{i}{j})_+ + (w_{ij}-\alf{j}{i})_+ < \max \bigg[(w_{ij}-\alf{i}{j})_+, (w_{ij}-\alf{j}{i})_+, 2w_{ij}-\alf{i}{j}-\alf{j}{i}\bigg]\leq w_{ij}$ which contradicts Lemma \ref{lemma:gamma_sat_dual_const}
that $\bargamma$ satisfies the constraints of the dual problem (\ref{prob:mwm_dual}).

Finally, we need to show that $\gamma_i=\off{j}{i}$ and $\gamma_j>\off{i}{j}$ give $\gamma_i=0$. First note that by equivalence of $(b)$ and $(c)$ we should have $w_{ij} < \alf{i}{j} + \alf{j}{i}$. On the other hand $\alf{i}{j}\leq \gamma_i  =\off{j}{i}\leq (w_{ij}-\alf{j}{i})_+$.  Now if  $w_{ij}-\alf{j}{i}>0$ we get $\alf{i}{j}\leq w_{ij}-\alf{j}{i}$ which is a contradiction. Therefore $\gamma_i=(w_{ij}-\alf{j}{i})_+=0$.
\enp
\end{proof}

\begin{lemma}
The following are equivalent:
\begin{itemize}
\item[(a)] $P(i) = \{ j \}$,
\item[(b)] $P(j) = \{ i \}$,
\item[(c)] $w_{ij}-\alf{i}{j} - \alf{j}{i} >0$.
\end{itemize}
\label{lemma:strong_dotted}
\end{lemma}
\begin{proof}
$(a)\Rightarrow (c) \Rightarrow (b)$: $(a)$ means $\alf{i}{j}<\off{j}{i}$. It also follows that $\off{j}{i}>0$ or
$(w_{ij}-\alf{j}{i})_+=w_{ij}-\alf{j}{i}$. Therefore, $\off{j}{i} \leq (w_{ij}-\alf{j}{i})_+ = w_{ij}-\alf{j}{i}$ which gives $w_{ij}-\alf{i}{j} - \alf{j}{i} >0$ or $(c)$.  From this we can explicitly write $\off{i}{j}=(w_{ij}-\alf{i}{j}+\alf{j}{i})/2$ which is strictly bigger than $\alf{j}{i}$. Hence we obtain $(b)$.

By symmetry $(b)\Rightarrow (c) \Rightarrow (c)$.  This finishes the proof.
\enp
\end{proof}

Recall that $(ij)$ is a weak-dotted edge if $w_{ij}-\alf{i}{j} - \alf{j}{i}=0$, a strong-dotted edge if
$w_{ij}-\alf{i}{j} - \alf{j}{i} > 0$, and a non-dotted edge otherwise.
Basically, for any dotted edge $(ij)$ we have $j\in P(i)$ and $i\in P(j)$.

\begin{corollary} One corollary of Lemmas \ref{lemma:i-jpartners_equivalence}-\ref{lemma:strong_dotted} is that strong-dotted edges are only adjacent to non-dotted edges.  Also each weak-dotted edge is adjacent to at least one weak-dotted edge at each end.
\end{corollary}

\begin{lemma}
If $i$ has no adjacent dotted edges, then $\gamma_i=0$
\label{lemma:no_dotted_means_gamma0}
\end{lemma}
\begin{proof}
Assume that the largest offer to $i$ comes from $j$. Therefore, $\alf{i}{j}\leq \off{j}{i}\leq (w_{ij}-\alf{j}{i})_+$.
Now if $w_{ij}-\alf{j}{i}>0$ then $\alf{i}{j}\leq w_{ij}-\alf{j}{i}$ or $(ij)$ is dotted edge which is impossible. Thus, $w_{ij}-\alf{j}{i}=0$ and $\gamma_i=0$.
\enp
\end{proof}

\begin{lemma}
The following are equivalent:
\begin{itemize}
\item[(a)] $\alf{i}{j}=\gamma_i$,
\item[(b)] $w_{ij}-\alf{i}{j} - \alf{j}{i} \le 0$,
\item[(c)] $\off{i}{j}=(w_{ij}-\alf{i}{j})_+$.
\end{itemize}
\label{lemma:when_alf_equals_gamma}
\end{lemma}
\begin{proof}
$(a)\Rightarrow (b)$: Follows from Lemma \ref{lemma:strong_dotted}, since $\alf{i}{j}=\gamma_i$ gives $|P(i)|>1$.

$(b)\Rightarrow (c)$: Follows from the definition of $\off{i}{j}$.

$(c)\Rightarrow (a)$: From $\off{i}{j}=(w_{ij}-\alf{i}{j})_+$ we have $w_{ij}-\alf{i}{j} - \alf{j}{i} \le 0$. Therefore, $\off{j}{i}=(w_{ij}-\alf{j}{i})_+\leq \max\big[w_{ij}-\alf{j}{i},0\big]\leq \alf{i}{j}$.
\enp
\end{proof}
Note that (c) is symmetric in $i$ and $j$, so (a) and (b) can be transformed by interchanging $i$ and $j$.
\begin{corollary}
$\alf{i}{j}=\gamma_i$ iff $\alf{j}{i}=\gamma_j$
\label{corr:alf_gamma_eq}
\end{corollary}

\begin{lemma}
$\off{i}{j}=(w_{ij}-\gamma_i)_+$ holds $\forall \; (ij) \in E$
\label{lemma:off_from_gamma}
\end{lemma}
\begin{proof}
If $w_{ij}-\alf{i}{j} - \alf{j}{i} \le 0$ then the result follows from Lemma \ref{lemma:when_alf_equals_gamma}.  Otherwise, $(ij)$ is strongly dotted and $\gamma_i=\off{j}{i}=(w_{ij}-\alf{j}{i} + \alf{i}{j})/2$,  $\gamma_j=\off{i}{j}=(w_{ij}-\alf{i}{j} + \alf{j}{i})/2$. From here we can explicitly calculate $w_{ij}-\gamma_i=(w_{ij}-\alf{i}{j} + \alf{j}{i})/2=\off{i}{j}$.
\enp
\end{proof}

\begin{lemma}
The unmatched balance property, equation \eqref{eq:balance}, holds at every edge $(ij) \in E$, and both sides of the equation
are non-negative.
\label{lemma:balance_at_FP}
\end{lemma}
\begin{proof}
In light of lemma \ref{lemma:off_from_gamma}, \eqref{eq:balance} can be rewritten at a fixed point as
\begin{align}
\gamma_i - \alf{i}{j} = \gamma_j - \alf{j}{i}
\label{eq:FP_balance}
\end{align}
which is easy to verify. The case $w_{ij}-\alf{i}{j} - \alf{j}{i} \le 0$ leads to both
sides of Eq.~(\ref{eq:FP_balance}) being $0$ by Corollary \ref{corr:alf_gamma_eq}. The
other case $w_{ij}-\alf{i}{j} - \alf{j}{i} >0$ leads to
\begin{align}
\off{i}{j}-\alf{j}{i}=\off{j}{i}-\alf{i}{j}=\frac{w_{ij}-\alf{i}{j} - \alf{j}{i} }{2}
\end{align}
Clearly, we have $\gamma_i=\off{j}{i}$ and $\gamma_j=\off{i}{j}$. So Eq.~(\ref{eq:FP_balance}) holds.
\enp
\end{proof}
Next lemmas show that dotted edges are in correspondence with the solid edges that were defined in Section \ref{sec:FixedPoint}.

\begin{lemma}
A non-solid edge cannot be a dotted edge, weak or strong.
\label{lemma:no_extra_dotted}
\end{lemma}

Before proving the lemma let us define alternating paths. A path $P=( {i_1}, {i_2},\ldots, {i_k})$ in $G$ is called \emph{alternating path} if:
(a) There exist a partition {of} edges of $P$ into two sets $A,B$ such that either $A\subset M^*$ or $B\subset M^*$.
Moreover $A$ ($B$) consists of all \emph{odd (even)} edges; i.e. $A=\{( {i_1}, {i_2}), ( {i_3}, {i_4}),\ldots\}$
($B=\{( {i_2}, {i_3}), ( {i_4}, {i_5}),\ldots\}$). (b) The path $P$ might intersect itself or even
repeat its own edges but no edge is repeated immediately.
That is, for any $1\leq r\leq k-2:~~~~ i_r\neq  i_{r+1}$ and $ i_r\neq  i_{r+2}$.  $P$ is called an \emph{alternating cycle} if $ {i_1}= {i_k}$.

Also, consider $\underline{x}^*$ and $\underline{y}^*$ that are optimum solutions for the LP and its dual, \eqref{prob:mwm_relaxation} and \eqref{prob:mwm_dual}. The complementary slackness conditions (see \cite{Schrijver}) for more details)
state that for all $v\in V$, $y_v^*(\sum_{e\in \p v}x_e^*-1)=$ and for all $e=(ij)\in E$, $x_e^*(y_i^*+y_j^*-w_{ij})=0$.
Therefore, for all solid edges the equality $y_i^*+y_j^*=w_{ij}$ holds. Moreover, any node $v\in V$ is adjacent to a solid edge iff $y_v^*>0$.

\begin{proof}[Proof of Lemma \ref{lemma:no_extra_dotted}]
We need to consider two cases:

\textbf{Case (I).} Assume that the optimum LP solution $\underline{x}^*$ is integer (having a tight LP).

The idea of the proof is that if there exist a non-solid edge which is dotted, we use a similar analysis to \cite{BayatiB} to construct an alternating path consisting of dotted and solid edges that leads to creation of at least two optimal solutions to LP \eqref{prob:mwm_relaxation}.  This contradicts with uniqueness assumption on the optimum solution of LP.

Now assume the contrary: take $(i_1,i_2)$ that is a non-solid edge but it is dotted. Consider an endpoint of $(i_1,i_2)$. For example take $i_2$. Either there is a solid edge attached to $i_2$ or not. If there is not, we stop. Otherwise, assume $(i_2,i_3)$ is a solid edge. Using Lemma \ref{lemma:no_dotted_means_gamma0}, either $\gamma_{i_3}=0$ or there is a dotted edge connected to $i_3$. But if this dotted edge is $(i_2,i_3)$ then $P(i_2)\supset\{i_1,i_3\}$. Therefore, by Lemma \ref{lemma:strong_dotted} there has to be another dotted edge $(i_3,i_4)$ connected to $i_3$. Now, depending on whether $i_4$ has (has not) an adjacent solid edge we continue (stop) the construction. Similar procedure could be done by starting at $i_1$ instead of $i_2$. Therefore, we obtain an alternating path $P=({i_{-k}},\ldots, i_{-1}, {i_0}, {i_1}, {i_2},\ldots, {i_\ell})$ with all odd edges being dotted and all even edges being solid. Using the same argument as in \cite{BayatiB} one can show that one of the following four scenarios occur.

\begin{itemize}
\item \textbf{Path:} Before $P$ intersects itself, both end-points of the path stop. Either the last edge is solid (then $\gamma_v =0$ for the last node) or the last edge is a dotted edg.  Now consider a new solution $\underline{x}'$ to LP \eqref{prob:mwm_relaxation} by $x_e'=x_e^*$ if $e\notin P$ and $x_e'=1-x_e^*$ if $e\in P$.  It is easy to see that $\underline{x}'$ is a feasible LP solution at all points $v\notin P$ and also for internal vertices of $P$. The only nontrivial case is when $v=i_{-k}$ (or $v=i_{\ell}$) and the edge $(i_{-k},i_{-k+1})$ (or $(i_{\ell-1},i_{\ell})$ ) is dotted. In both of these cases, by construction $y_v^*=0$ which means no solid edge is attached to $v$ outside of $P$ so making any change inside of $P$ is safe. Now denote the weight of all solid (dotted) edges of $P$ by $w(P_{\textrm{solid}})$ ($w(P_{\textrm{dotted}})$). Hence,
    \begin{align} \sum_{e \in E} w_e x_e^*- \sum_{e \in E} w_e x_e'=w(P_{\textrm{solid}})-w(P_{\textrm{dotted}}).
    \label{eq:w(solid)-w(dotted)}
    \end{align}
    But $w(P_{\textrm{dotted}})=\sum_{v\in P}\gamma_v$ .  Moreover, from Lemma \ref{lemma:gamma_sat_dual_const},  $\underline{\gamma}$ is dual feasible which gives $w(P_{\textrm{solid}})\leq\sum_{v\in P}\gamma_v$. We are using the fact that if there is a solid edge at an endpoint of $P$ the $\gamma$ of the endpoint should be $0$. Now \eqref{eq:w(solid)-w(dotted)} reduces to $w_e x_e^*- \sum_{e \in E} w_e x_e'\leq 0.$  This contradicts the tightness of LP relaxation \eqref{prob:mwm_relaxation} since $x_e'\neq x_e^*$ holds at least for $e=(i_1,i_2)$.

\item \textbf{Cycle:} $P$ intersects itself and will contain an even cycle $C_{2s}$. This case can be handled very similar to the path by defining $x_e'=x_e^*$ if $e\notin C_{2s}$ and $x_e'=1-x_e^*$ if $e\in C_{2s}$. The proof is even simpler since the extra check for the boundary condition is not necessary.
\item \textbf{Blossom:} $P$ intersects itself and will contain an odd cycle $C_{2s+1}$ with a path (stem) $P'$ attached to the cycle at point $u$. In this case let $x_e'=x_e^*$ if $e\notin P'\cup C_{2s+1}$, and $x_e'=1-x_e^*$ if $e\in P'$, and $x_e'=\frac{1}{2}$ if $e\in C_{2s+1}$. From here, we drop the subindex $2s+1$ to simplify the notation. Since the cycle has odd length, both neighbors of $u$ in $C$ have to be dotted.
    Therefore,
\begin{align} \sum_{e \in E} w_e x_e^*- \sum_{e \in E} w_e x_e'&=w(P'_{\textrm{solid}})+w(C_{\textrm{solid}})-w(P'_{\textrm{dotted}})-\frac{w(C_{\textrm{dotted}})+w(C_{\textrm{solid}})}{2}\\
x_e'&=w(P'_{\textrm{solid}})+\frac{w(C_{\textrm{solid}})}{2}-w(P'_{\textrm{dotted}})-\frac{w(C_{\textrm{dotted}})}{2}\\
&\leq \sum_{v\in P'}\gamma_v + \frac{\sum_{v\in C}\gamma_v-\gamma_u}{2} - (\sum_{v\in P'}\gamma_v -\gamma_u) - \frac{\sum_{v\in C}\gamma_v+\gamma_u}{2}\leq 0,
    \label{eq:w(solid)-w(dotted)-blossom}
    \end{align}
    which is again a contradiction.
\item \textbf{Bicycle:} $P$ intersects itself at least twice and will contain two odd cycles $C_{2s+1}$ and $C'_{2s'+1}$ with a path (stem) $P'$ that is connecting them. Very similar to Blossom, let $x_e'=x_e^*$ if $e\notin P'\cup C\cup C'$, $x_e'=1-x_e^*$ if $e\in P'$, and $x_e'=\frac{1}{2}$ if $e\in C\cup C'$. The proof follows similar to the case of blossom.
\end{itemize}

\textbf{Case (II).} Assume that the optimum LP solution $\underline{x}^*$ is not necessarily integer.

Everything is similar to Case (I) but the algebraic treatments are slightly different. Some edges $e$ in $P$ can be $\f{1}{2}$-solid ($x_e^*=\frac{1}{2}$). In particular some of the odd edges (dotted edges) of $P$ can now be $\f{1}{2}$-solid. But the subset of $\f{1}{2}$-solid edges of $P$ can be only sub-paths of odd length in $P$. On each such sub-path defining $\underline{x}'=1-\underline{x}^*$ means we are not affecting $\underline{x}^*$. Therefore, all of the algebraic calculations should be considered on those sub-paths of $P$ that have no $\f{1}{2}$-solid edge which means both of their boundary edges are dotted.

\begin{itemize}
\item \textbf{Path:} Define $\underline{x}'$ as in Case (I). Using the discussion above, let $P^{(1)},\ldots,P_{(r)}$ be disjoint sub-paths of $P$ that have no $\f{1}{2}$-solid edge. Thus,
    \[\sum_{e \in E} w_e x_e^*- \sum_{e \in E} w_e x_e'=\sum_{i=1}^r \bigg[w(P_{\textrm{solid}}^{(i)})-w(P_{\textrm{dotted}}^{(i)})\bigg].\]
    Since in each $P^{(i)}$ the two boundary edges are dotted,  $w(P_{\textrm{solid}}^{(i)})\leq \sum_{v\in P^{(i)}}\gamma_v$ and $\sum_{v\in P^{(i)}}\gamma_v=w(P_{\textrm{dotted}}^{(i)})$. The rest can be done as in Case (I).
\item \textbf{Cycle, Blossom, Bicycle:} These cases can be done using the same method of breaking the path and cycles into sub-paths $P^{(i)}$ and following the case of path.
\end{itemize}

\end{proof}

\begin{lemma}
Every $1$-solid edge is a strong-dotted edge. Also, every $\frac{1}{2}$-solid edge is a weak-dotted edge.
\label{lemma:no_extra_solid}
\end{lemma}
\begin{proof}
From Lemma \ref{lemma:no_extra_dotted} it follows that the set of dotted edges is a subset of the solid edges. In particular, no node can be adjacent to more than two dotted edges. Now using Lemma \ref{lemma:strong_dotted} the set of dotted edges is a disjoint union of isolated edges (strongly dotted) and cycles (weakly dotted edges). If we define a $\underline{x}'$ to be zero on all non-dotted edges and $x_e'=1$ when $e$ is strongly dotted and $x_e'=\frac{1}{2}$ for weakly dotted ones. Then clearly $\underline{x}'$ is feasible to \eqref{prob:mwm_relaxation}.  On the other hand using Lemma \ref{lemma:no_dotted_means_gamma0} we have $\sum_{e \in E} w_e x_e'=\sum_{v\in V}\gamma_v$.  But $\bargamma$ is dual feasible which means $\sum_{v\in V}\gamma_v\geq \sum_{v\in V}y_v^*=\sum_{e \in E} w_e x_e^*$ which shows that $\underline{x}'$ is also an optimum solution to \eqref{prob:mwm_relaxation}. By uniqueness assumption on $\underline{x}^*$ we obtain the desired result.
\end{proof}

\begin{lemma}
$\bargamma$ is an optimum for the dual problem (\ref{prob:mwm_dual})
\label{lemma:gamma_opt}
\end{lemma}
\begin{proof}
Lemma \ref{lemma:gamma_sat_dual_const} guarantees feasibility.
Optimality follows from lemmas \ref{lemma:no_dotted_means_gamma0}, \ref{lemma:no_extra_dotted}
and \ref{lemma:no_extra_solid}.
\end{proof}

%
%
\subsection{Fixed points and the KT sequence}

\begin{lemma}
Let $G$ be an instance admitting a NB solution $\baralf^*$ with
gap $\sigma$, and KT sequence $(\cC_0,\cC_1,\cC_2,\dots,\cC_k)$.
Denote by  $M^*$ the max weight
matching of $G$.
Then for any
$q\in\{1,\dots,k\}$ and any edge $(i,j)\in E(\cC_q),
$
\begin{eqnarray}
\alf{i}{j}^*+\alf{j}{i}^*-w_{ij} = \left\{
\begin{array}{ll}-2\sigma_q&\mbox{ if $(i,j)\in E_1(\cC_q)$,}\\
\sigma_q &\mbox{ if $(i,j)\in E_2(\cC_q)$}\\
\end{array}\right.
\label{eq:KTstruct_fp}
\end{eqnarray}
Further, for any edge $(i,j)\not\in \cC_q$ for any $q$,
s.t. $i\in \cC_{q(i)}$, $j\in \cC_{q(j)}$, we have
\begin{eqnarray}
\alf{i}{j}^*+\alf{j}{i}^*-w_{ij} \ge \max(\sigma_{q(i)},\sigma_{q(j)})\, .
\label{eq:edge_different_l}
\end{eqnarray}
\end{lemma}
\begin{proof}
Consider $(i,j) \in E(\cC_q)$.

Suppose $(i,j) \in M^*$, i.e. $(i,j) \in E_1(\cC_q)$. We know that $ i,j \in V(\cC_q)$. Hence the node
slacks for each of $i$ and $j$ are $\sigma_q$. Hence, $\gamma_j^*-\alf{j}{i}^*=\gamma_i^* -\alf{i}{j}^*=\sigma_q$.
Also, $\gamma_i^*+\gamma_j^*=w_{ij}$. Hence $\alf{j}{i}^*+\alf{i}{j}^*-w_{ij}= -2\sigma_q$

Suppose $(i,j) \notin M^*$, i.e. $(i,j) \in E_2(\cC_q)$.
We know that $\alf{i}{j}^*=\gamma_i,\; \alf{j}{i}^*=\gamma_j$.
(\ref{eq:edge_slack}) now yields $\alf{i}{j}^*+\alf{i}{j}^*-w_{ij}=\sigma_q$ as required.

Now consider any edge $(i,j)\not\in \cC_q$ for any $q$. Consider $\off{j}{i}^*$.
We must have $\off{j}{i}^* \leq	\gamma_i^*-\sigma_{q(i)} = \alf{i}{j}^*-\sigma_{q(i)}$.
We know that $\alf{j}{i}^*+\alf{i}{j}^*-w_{ij} \geq 0$ so $\off{j}{i}^* \geq w_{ij}-\alf{j}{i}^*$.
Combining, we have $\alf{j}{i}^*+\alf{i}{j}^*-w_{ij} \geq \sigma_{q(i)}$. Similarly,
$\alf{j}{i}^*+\alf{i}{j}^*-w_{ij} \geq \sigma_{q(j)}$. This completes the proof.

Note: If we include the condition given by Eq.~(\ref{eq:extra_sigma_condition}) in the definition
of $\sigma$, we have the stronger condition
$\alf{j}{i}^*+\alf{i}{j}^*-w_{ij} \geq \sigma_{q(i)} + \sigma$ when $q(i)=q(j)\geq 1$.
\end{proof}

%
%
\section{Proof of convergence lemmas}
In this section we first prove some basic properties of the natural dynamics in Section \ref{sec:BasicConvergence}.  Then in Sections \ref{app:Path}-\ref{app:Bicycle} the lemmas that are used for the proofs of convergence on basic structures (path, blossom, and bicycle) are proved.

Throughout these sections, we say that an alternating path or blossom is
\textit{anchored}  \cite{KT} at its degree-1 node(s). Note that
$V_{\textup{ext}}(\cC_q) \ V(\cC_q)$ \textit{may} contain anchored node(s),
but no other nodes.

\subsection{Basic properties}
\label{sec:BasicConvergence}

\begin{claim}
Consider message vectors $\baralf$ and $\barbet$. Suppose, for $(ij) \in E$,
$|\alf{i}{j}-\bet{i}{j}| \leq \Delta$ and $|\alf{j}{i}-\bet{j}{i}| \leq \Delta$. We have
\begin{align}
\left|\;\off{i}{j}|_{\baralf}-\off{i}{j}|_{\barbet} \;\right| \leq \Delta
\label{eq:offer_non_expansion}
\end{align}
where $\off{i}{j}|_{\baralf}$ ($\off{i}{j}|_{\barbet}$) refers to the offers corresponding to message vector $\baralf$ ($\barbet$).
\label{claim:offer_non_expansion}
\end{claim}
\begin{proof}
This follows from definition \ref{eq:off_def}. $\off{i}{j}=f(\alf{i}{j}, \alf{j}{i})$,
where $f(x,y):\reals_+^2 \rightarrow \reals_+$ is defined as
\begin{eqnarray}
f(x,y) = \left\{
	\begin{array}{ll} \frac{w_{ij}-x+y}{2} & x+y \leq w_{ij}\\
	(w_{ij}-x)_+ &\mbox{ otherwise.}\\
\end{array}\right.
\label{eq:regionwise_offers}
\end{eqnarray}
It is easy to check that $f$ is continuous everywhere in $\reals_+^2$.
Also, it is differentiable except in $\{(x,y) \in \reals_+^2:x+y=w_{ij} \mbox{ or } x=w_{ij}\}$, and
satisfies $||\nabla f ||_1=|\frac{\partial f}{\partial x}|+|\frac{\partial f}{\partial y}| \leq 1$
Hence, $f$ is Lipschitz continuous in the $L^\infty$ norm, with Lipschitz constant 1,
leading to (\ref{eq:offer_non_expansion}).
\end{proof}

\begin{proof}[Proof of Lemma \ref{lem:infty_non_expansion}]
Use Claim \ref{claim:offer_non_expansion} with $\Delta = ||\baralf^0 - \barbet^0||_\infty$ for
every offer $\off{i}{j}^0$ in the graph. The result follows from the update rule (\ref{eq:update}).
\end{proof}

\begin{definition}
We overload $\Bo$ and say that $\Bo(q, \delta)$ holds for $\baralf$ if
$U_{\cG,\cG_q}(\baralf)\le \Delta-\delta\,,\; U_{\cG}(\baralf)\le \Delta\,$
\label{def:Bo_overloaded}
\end{definition}

\begin{proof}[Proof of Lemma \ref{lemma:outsideOff_ub}]
Note that $t$ plays no role in this result. As such we can restate it in light of
definition \ref{def:Bo_overloaded} as:
``If $\Bo(q,\delta)$ holds for $\baralf$ then for all $(i,j) \in E$ with $i$ in $V(\cC_q)$ and $(i,j)\notin E(\cC_q)$:
$\off{j}{i} \leq \gamma_i^* -\sigma_q+\Delta -\delta$." We prove this below.

For $q=0$, we have $\gamma_i^*=\sigma_0=0$.
Consider $i \in V(\cC_0)$. Consider any $j \in \partial i$. We must have $j \in \cC_l$ for some
$1 \leq l \leq k$. Also, $\alf{i}{j}^*=0$. Hence, by Eq.~(\ref{eq:edge_different_l}), we know that
$\alf{j}{i}^* \geq w_{ij}+\sigma_l \geq w_{ij}+\sigma$. We have $\alf{j}{i} \geq \alf{j}{i}^* - \Delta$. Hence,
$\off{j}{i} \leq (\Delta - \sigma)_+ \leq (\Delta -\delta)$ as needed.

For $q>0$ we consider 3 cases. We make use of the fact that $(i,j) \notin M^* \Rightarrow \alf{i}{j}^* = \gamma_i^*,
\alf{j}{i}^*=\gamma_j^*$, $\off{j}{i} \leq (w_{ij}- \alf{j}{i})_+$.

Case (i): $j \in V(\cG_q)$\\
By Eq.~(\ref{eq:edge_different_l}), we know that $\alf{i}{j}^* +\alf{j}{i}^* \geq w_{ij}+\sigma_q$.
It follows that
\begin{align*}
\off{j}{i}\leq (w_{ij}-\alf{j}{i}^*+\Delta -\delta)_+ \leq  (\gamma_i^* -\sigma_q+\Delta -\delta)_+=\gamma_i^* -\sigma_q+\Delta -\delta
\end{align*}
as required.

Case (ii): $j \in V(\cC_q)$\\
We use condition (\ref{eq:extra_sigma_condition}) which leads to $\gamma_i^* +\alf{j}{i}^* \geq w_{ij}+\sigma_q+\sigma$.
Hence
\begin{align}
\off{j}{i}\leq (w_{ij}-\alf{j}{i}^*+\Delta)_+ \leq  (\gamma_i^* -\sigma_q+\Delta -\sigma)_+\leq\gamma_i^* -\sigma_q+\Delta -\delta
\label{eq:outsideOff_ub_case2}
\end{align}

Case (iii): $j \notin V(\cG_{q+1})$\\
We know, by Lemma \ref{eq:edge_different_l} that $\alf{i}{j}^*+\alf{j}{i}^* \geq w_{ij}+\sigma_q+\sigma$.
(\ref{eq:outsideOff_ub_case2}) follows.
\end{proof}

\begin{proof}[Proof of Lemma \ref{lemma:CondStrucPersistence}]
We show that $\Bo(q,\delta,1)$ holds. Then by induction on $t$, $\Bo(q,\delta,t)$ holds for all times $t>0$.

Consider $q=0$. Suppose $\Bo(0,\delta,0)$ holds i.e. $U_{\cG}(\baralf^0)\le \Delta$.
This leads to $\Bo(0,\delta,1)$ as a consequence of Lemma \ref{lem:infty_non_expansion} with $\barbet^0=\baralf^*$.

Consider $q=1$. Suppose $\Bo(1,\delta,0)$ holds. Take any $i \in \cC_0$.
By Lemma \ref{lemma:outsideOff_ub}, for any $j \in \partial i$ we know that
$\off{j}{i}^0 \leq (\Delta -\delta)$. It follows that for any $k \in \partial i $,
$|\alf{i}{k}^1-\alf{i}{k}^*|=\alf{i}{k}^1 \leq \Delta - \delta$, since we already had
$|\alf{i}{k}^0-\alf{i}{k}^*|=\alf{i}{k}^0 \leq \Delta - \delta$.
i.e. $\Bo(1,\delta,1)$ holds.

We now proceed by induction on $q$. Suppose we know that $\Bo(Q,\delta,0) \Rightarrow \Bo(Q,\delta,1), 0 < Q \leq k$.
Suppose $\Bo(Q+1,\delta,0)$ holds. In particular $\Bo(Q,\delta,0)$ holds, and hence $ \Bo(Q,\delta,1)$ also holds.
Consider $i \in \cC_{Q}$. We know $i \in M^*$, so there is $\parti \mbox{ s.t. } (i,\parti) \in M^*$.
Define $S=\{j:(i,j) \in E(\cC_Q), j \neq \parti\}$. Then
$|S| \in \{0,1,2\}$. We know that $\off{\parti}{i}^*=\gamma_i^*$ and $\off{j}{i}^* = \gamma_i^*-\sigma_Q \
\forall j \in S$. Denote this `second-best' offer by $\psi_i^*$, with $\psi_i^*=0$ if $|S|=0$.

Take any edge $(l,i) \in E$.

Case (i): $l \in V(\cG_{Q+1})$\\
By assumption, $|\alf{l}{i}^0-\alf{l}{i}^*|,|\alf{i}{l}^0-\alf{i}{l}^*|  \leq \Delta -\delta$.
We know, by Lemma \ref{claim:offer_non_expansion} that
\begin{align}
|\off{l}{i}^0-\off{l}{i}^*| \leq \Delta -\delta
\label{eq:type1_offers}
\end{align}

Case (ii): $l \notin V(\cG_{Q+1})$\\
Lemma \ref{lemma:outsideOff_ub} is applicable (same as Case (iii) in the proof of Lemma \ref{lemma:outsideOff_ub}). Hence,
\begin{align}
\off{l}{i}^0 \leq \psi_i^* + \Delta - \delta
\label{eq:type2_offers}
\end{align}

With (\ref{eq:type1_offers}) and (\ref{eq:type2_offers}) established, we now check
that $\forall l \in \partial i$
\begin{align}
|\alf{i}{l}^1-\alf{i}{l}^*| \leq \Delta -\delta \
\label{eq:persistence_newmessage_bound}
\end{align}

$\alf{i}{\parti}^*=\psi_i^*$. We verify separately for each of the cases $|S|=0$ and $|S|>0$
that
\begin{align}
|\max_{j \in \partial i \backslash \parti}\off{j}{i}^0  - \psi_i^*| \leq \Delta -\delta
\label{eq:persistence_nonpartner_message_bound}
\end{align}
Combine
with $|\alf{i}{\parti}^0-\alf{i}{\parti}^*| \leq \Delta -\delta$ to get
(\ref{eq:persistence_newmessage_bound}) for $l=\parti$.

If $l \in \partial i \backslash \parti$, $\alf{i}{l}^*=\gamma_i^*$. We know that
\begin{align}
|\off{\parti}{i}^0 - \off{\parti}{i}^*|=|\off{\parti}{i}^0 - \gamma_i^*|\leq (\Delta -\delta)
\label{eq:persistence_partner_message_bound}
\end{align}
Combining (\ref{eq:persistence_nonpartner_message_bound}) and
(\ref{eq:persistence_partner_message_bound}) leads to
\begin{align}
|\max_{j \in \partial i \backslash l}\off{j}{i}^0  - \gamma_i^*| \leq \Delta -\delta
\end{align}
and hence to (\ref{eq:persistence_newmessage_bound}) for $l \neq \parti$.

Hence, $\Bo(Q+1, \delta, 1)$ holds as required. Induction on $Q$ completes the proof.
\end{proof}

\begin{proof}[Proof of Lemma \ref{lemma:outgoing_error_bound}]
By Lemma \ref{lemma:CondStrucPersistence}, we know that $\Bo(q,\delta,t)$ holds for all $t \geq 0$.

Consider $q=0$. Take any $i \in V(\cC_0)$. By Lemma \ref{lemma:outsideOff_ub}, for any $\parti \in \partial i$ we know that
$\off{\parti}{i}^t \leq (\Delta -\delta)$. Hence, $\max_{\parti \in \partial i \backslash j}\off{\parti}{i}^t \leq (\Delta -\delta)$
for any $j \in \partial i, \ \forall \, t \geq 0$. The result follows from $\alf{i}{j}^0 \leq \Delta$ and the update
rule \ref{eq:update}.

Consider $q>0$. Take any $i \in V(\cC_q)$. We know that $(i,\parti) \in M^*$ for some $\parti \in  V(\cC_q)$.
It follows from the assumption that $U_{G,\cC_q}(\baralf^t)\leq \Delta-\delta$ that $|\off{\parti}{i}^t-\gamma_i^*| \leq \Delta -\delta$.
For any $j \in \partial i \backslash \parti$, we know that
$\off{j}{i}^t \leq \off{j}{i}^*+\Delta \leq \gamma_i^*-\sigma_q+\Delta \leq \gamma_i^*+\Delta -\delta$.
Hence, $ |\max_{j \in \partial i \backslash l}\off{j}{i}^t-\alf{i}{l}^*|=
|\max_{j \in \partial i \backslash l}\off{j}{i}^t-\gamma_i^*|
\leq (\Delta -\delta)$
for any $l \in \partial i \backslash i_\P$. The result follows.

\end{proof}

\subsection{Path convergence}
\label{app:Path}
\begin{definition} Recall the path structures from Section \ref{subsubsec:path_descent}.
\begin{enumerate}
\item We use $P_\len$ to denote a path $0-1-\ldots-\len$.
\item  We call a matching $M_P$ on a path $P$ with alternate edges included a `valid path  matching'.
\end{enumerate}
\end{definition}

\begin{lemma}\label{lemma:RW}
Let $P=P_\len, \; \len\geq 1$. Suppose we apply arbitrary boundary conditions
$\{\boundary_\L^t\}_{t\geq 0}, \{\boundary_\R^t\}_{t\geq 0}$.
Consider the simplified dynamics on $P_\len$ wrt given boundary conditions, a
valid path matching $M_P$ and
two different initial conditions $\sbaralf^0$ and $\sbarbet^0$. We have,

\begin{eqnarray}
||\sbaralf^t-\sbarbet^t||_{\infty}\le
C \len \exp\left(-\frac{t}{C\len^2}\right)||\sbaralf^0-\sbarbet^0||_{\infty}\, ,
\label{eq:simplified_contraction}
\end{eqnarray}
for some $C>0$.
Also,
\begin{align}
||\sbaralf^t-\sbarbet^t||_{\infty}\le ||\sbaralf^0-\sbarbet^0||_{\infty}
\label{eq:simp_infty_non_expansion}
\end{align}
\label{lemma:simplified_contraction}
\end{lemma}
\begin{proof}
Consider the difference $\sbarDelta^t=\sbaralf^t-
\sbarbet^t$. We know that $||\sbarDelta^0||_\infty =||\sbaralf^0-\sbarbet^0||_{\infty}$.
We bound $\sbarDelta^t$  by the `mass' in an appropriately defined random walk on the path.

Firstly, we write update equations for $\sbarDelta$. We distinguish between the two `end' edges that
are at the ends of the path and the other `interior' edges. For an interior non-matching edge $(i,i+1) \notin M_P$,
\begin{align}
\sDel{i}{i+1}^{t+1}&=\damp\left(\frac{ -\sDel{i-1}{i}^t + \sDel{i}{i-1}^t }{2}\right)
+ (1-\damp)\sDel{i}{i+1}^t
\nonumber \\
\sDel{i+1}{i}^{t+1}&=\damp\left(\frac{ -\sDel{i+2}{i+1}^t + \sDel{i+1}{i+2}^t }{2}\right)
+ (1-\damp)\sDel{i+1}{i}^t
\label{eq:matching_simple_path}
\end{align}

For an interior matching edge $(i,i+1) \in M_P$
\begin{align}
\sDel{i}{i+1}^{t+1}&= -\damp\sDel{i-1}{i}^t
+ (1-\damp)\sDel{i}{i+1}^t \nonumber \\
\sDel{i+1}{i}^{t+1}&= -\damp\sDel{i+2}{i+1}^t
+ (1-\damp)\sDel{i+1}{i}^t
\label{eq:nonmatch_simple_path}
\end{align}

For the end edges, we have $\sDel{0}{1}^{t+1}=\damp \sDel{0}{1}^t\, , \sDel{\len}{\len-1}^{t+1}=\damp \sDel{\len}{\len-1}^t\, \ \forall t \geq 0$.
In the other direction, i.e. for $\sDel{1}{0}^{t+1}$, and $\sDel{\len-1}{\len}^{t+1}$ the relevant update
rule from Eqs.~(\ref{eq:matching_simple_path}) and (\ref{eq:nonmatch_simple_path}) is
applicable, depending on the type of edge.

Now, define $\mass{i}{i+1}^0= \mass{i+1}{i}^0=|| \sbarDelta^0 ||_{\infty} \ \forall i \in \{0,1,\ldots,\len-1\}$.
Update $\barmass$ as follows. For an interior non-matching edge $(i,i+1) \notin M_P$,
\begin{align}
\mass{i}{i+1}^{t+1}&=\damp\left(\frac{ \mass{i-1}{i}^t + \mass{i}{i-1}^t }{2}\right)
+ (1-\damp)\mass{i}{i+1}^t \nonumber \\
\mass{i+1}{i}^{t+1}&=\damp\left(\frac{ \mass{i+2}{i+1}^t + \mass{i+1}{i+2}^t }{2}\right)
+ (1-\damp)\mass{i+1}{i}^t
\label{eq:matching_simple_path_bound}
\end{align}

For an interior matching edge $(i,i+1) \in M_P$
\begin{align}
\mass{i}{i+1}^{t+1}&= \damp\mass{i-1}{i}^t
+ (1-\damp)\mass{i}{i+1}^t \nonumber \\
\mass{i+1}{i}^{t+1}&= \damp\mass{i+2}{i+1}^t
+ (1-\damp)\mass{i+1}{i}^t
\label{eq:nonmatch_simple_path_bound}
\end{align}

For the end edges, we define $\mass{0}{1}^{t+1}=\damp \mass{0}{1}^t, \mass{\len}{\len-1}^{t+1}=\damp \mass{\len}{\len-1}^t, \ \forall t \geq 0$.
In the other direction, i.e. for $\mass{1}{0}^{t+1}$, and $\mass{\len-1}{\len}^{t+1}$ the relevant update
rule from Eqs.~(\ref{eq:matching_simple_path_bound}) and (\ref{eq:nonmatch_simple_path_bound}) is
applicable, depending on the type of edge.

\begin{claim}
$\mass{i}{i+1}^t \geq |\sDel{i}{i+1}^t|, \mass{i+1}{i}^t \geq |\sDel{i+1}{i}^t|  \ \forall i \in \{0,1,\ldots,\len-1\} \textup{ and } t \geq 0$
\label{claim:RW_bound}
\end{claim}
\begin{proof}
Clearly, the claim is true at $t=0$. An elementary proof by induction
can be constructed based on $|x+y| \leq |x|+|y|$. 
\end{proof}

In light of claim \ref{claim:RW_bound}, it is sufficient to prove that
$||\barmass^t||_\infty \leq C\, \len \exp(-t/C\, \len^2)|| \sbarDelta^0 ||_{\infty}$.
Now, for a normalization $Z = 2\len|| \sbarDelta^0 ||_{\infty}$,
$\barmass^t/Z$ can be thought of as the probability distribution of
a  particle performing the following random walk.
The particle is positioned on an edge $\{1,\dots,\len\}$
and has a direction (right or left). It starts at a uniformly random edge,
with uniformly random direction.
If the edge is non-matching, the particle keeps its direction
and moves along it  to the next edge, with probability $\damp$.
If the edge is matching, with probability $\damp/2$ it moves forward
maintaining direction, with probability $\damp/2$ it reverses direction
and moves forward one step, and with probability $(1-\damp)$
retains its status.
The particle is disappears when crossing any of the ends of the path.

Notice that the particle position, when watched on
matching edges, is a lazy random walk
(i.e. a simple random walk that keeps its position with probability
bounded away from $0$ and $1$), killed at $0$ and $\len'$,
with $\len'\le (\len+1)/2$. Hence
$||\barmass^t||_\infty/Z$ is not larger than the survival probability of
such  random walk, which is bounded as in the claim.

It is easy to check that Claim \ref{claim:offer_non_expansion} holds also for the simplified dynamics.
(\ref{eq:simp_infty_non_expansion}) follows.
\end{proof}
%
\begin{proof}[Proof of Lemma \ref{lemma:simplified_convergence}]
This follows directly from (\ref{eq:simplified_contraction}) in Lemma \ref{lemma:RW}.
\end{proof}

\begin{lemma}
Let $G$ be an instance and $P\subseteq G$ be a path in $G$.
Consider two initial conditions $\baralf^0$, $\barbet^0$ on $G$,
such that $\baralf^0_{G\setminus P} = \barbet^0_{G\setminus P}$
(they coincide outside $P$) and
$\baralf^0_P\succeq\barbet^0_P$ ($\baralf$ dominates $\barbet$ on $P$).
Then, after one iteration  $\baralf^1_P\succeq\barbet^1_P$.
\label{lemma:monotonicity}
\end{lemma}
\begin{proof}[Proof of Lemma \ref{lemma:monotonicity}]
Consider a node $i$ on $P$ with $i$ even. We compare the offers received by $i$
under $\baralf$ and $\barbet$ at time $0$. Clearly, all offers coming in from
$G-P$ are identical. If $i>0$ and $i$ even, by condition (\ref{def:ordering_odd}) and the fact
that $\off{i-1}{i}$ is non-increasing in $\alf{i-1}{i}$ and non-decreasing in
$\alf{i}{i-1}$, we know that $\off{i-1}{i}|_{\baralf^0} \geq \off{i-1}{i}|_{\barbet^0}$.
Hence, $\alf{i}{i+1}^1 \geq \bet{i}{i+1}^1$ as required. We can argue similarly for all
other types of messages on $P$.
\end{proof}

\begin{lemma}
Consider the simplified dynamics on a path instance $G=P_\len$ wrt a valid
path matching $M_P$. Let $\sbaralf^t$ denote the message vector resulting at time
$t \geq 0$ with initialization $\sbaralf^0$ and boundary conditions
$\{\boundary_{ 1,\L}^t\}_{t\geq 0}, \{\boundary_{ 1,\R}^t\}_{t\geq 0}$.
Similarly, let $\sbarbet^t$ denote the message vector resulting at time
$t \geq 0$ with initialization $\sbarbet^0$ and boundary conditions
$\{\boundary_{2,\L}^t\}_{t\geq 0}, \{\boundary_{2,\R}^t\}_{t\geq 0}$.
 If $\sbaralf^0 \succeq \sbarbet^0$ and
 \begin{align*}
 \boundary_{1,\L}^t\geq \boundary_{2,\L}^t \,,\qquad
  (-1)^\len\boundary_{1,\R}^t\geq (-1)^\len\boundary_{2,\R}^t \,,
  \qquad \forall \, t \geq 0
  \end{align*}
   then $\sbaralf^t \succeq \sbarbet^t \ \forall t \geq 0$.
\label{lemma:monotonicity_simplified}
\end{lemma}
\begin{proof}[Proof of Lemma \ref{lemma:monotonicity_simplified}]
We show that $\sbaralf^1 \succeq \sbarbet^1$. This follows for `internal'
messages from the fact that $\soff{i}{j}$ defined in \ref{eq:simplified_offers}
is monotonic non-increasing in $\salf{i}{j}$ and non-decreasing in $\salf{j}{i}$.
$\boundary_{1,\L}^0\geq \boundary_{2,\L}^0$ ensures that $\salf{0}{1}^1\geq \sbet{0}{1}^1$ and
similarly for the other boundary message.

Induction on $t$ completes the proof of the quoted result.
\end{proof}

Let $\barsoff^t(s)$ denote the offers made along edges of $P$ under
the bounding process $\barbound^t(s), \ t \geq 0$.

As discussed in Section \ref{subsubsec:path_descent} we have the following unique fixed point for the bounding process:
\begin{eqnarray}
\bound{i}{i+1}^* = \alf{i}{i+1}^*+s(-1)^i(\Delta-\delta)\, ,\;\;\;\;\;\;\;
\bound{i+1}{i}^* = \alf{i+1}{i}^*-s(-1)^{i}(\Delta -\delta)\, ,
\label{eq:bound_fp}
\end{eqnarray}
for all $i$ in $\{0,1,\ldots,\len-1\}$.

Clearly, $\barbound^*(+) \preceq \barbound^0(+)$. Hence, by Lemma
\ref{lemma:monotonicity_simplified} with $\boundary_{1,\L}^t=\boundary_{2,\L}^t$,
we know that $\barbound^*(+) \preceq \barbound^t(+)$ for all $t \geq 0$.
Also, from (\ref{eq:simp_infty_non_expansion}) we have
$||\barbound^t-\barbound^*||_\infty \leq ||\barbound^0-\barbound^*||_\infty=\delta$
leading to $\barbound^t(+) \preceq \barbound^0(+)$. Combining, for all $t\geq 0$ we have
\begin{align}
\barbound^*(+) \preceq \barbound^t(+) \preceq \barbound^0(+).
\label{eq:bounding_range}
\end{align}

\begin{lemma}\label{lemma:BoundConv}
Let $\{\barbound^t\}_{t\ge 0}$ be the $(s,\Delta,\delta)$-bounding
process on $P$. There exists a finite
$c_\RW$ such that for any $\eps \in (0,\delta)$,
for $t> c_\RW \len^{2.1}\left(1+\log(\delta/\eps)\right)$,
\begin{eqnarray}
||\barbound^t-\baralf^*||_{\infty} \le \Delta-\delta+\eps\, .
\end{eqnarray}
\end{lemma}
\begin{proof}
$||\barbound^0-\barbound^*||_\infty = \delta$. Hence, by
Lemma \ref{lemma:RW} a finite $c$ exists s.t.
\begin{align}
||\barbound^t-\barbound^*||_\infty \leq \epsilon \qquad \forall t \geq c_\RW \len^{2.1}\left(1+\log(\delta/\eps)\right).
\end{align}
Also, $||\barbound^*-\baralf^*||_\infty = \Delta -\delta$.
The result follows.
\end{proof}
\begin{proof}[Proof of Lemma \ref{lemma:CorrectMatching}]
It follows from (\ref{eq:bounding_range}) that for any $(i,j)\in E_P$ and all $t\geq 0$,
\begin{align}
|(\bound{i}{j}^t+\bound{j}{i}^t) - (\alf{i}{j}^*+\alf{j}{i}^*)|.
\label{eq:bound_delta_close}
\end{align}
Combining Eqs.~(\ref{eq:bound_delta_close}), and (\ref{eq:KTstruct_fp}), and
noting that $\delta \le \sigma \le \sigma_l$ we obtain our result.
\end{proof}

\begin{proof}[Proof of Lemma \ref{lemma:BoundingBounds}]
We proceed by induction. Clearly, Eq.~(\ref{eq:real_path_bound}) holds at
$t=0$. Suppose it holds at time $t$. We prove $\baralf^{t+1}_P\le \barbound^{t+1}(+)$ here, the other
inequality follows similarly. We drop $s=(+)$ from the notation henceforth for convenience.

Define $\tbarbound^t$ by $\tbound{i}{j}^t = \max\big( \min(\bound{i}{j}^t,W),0\big)$ for all $(i,j) \in E_P$ i.e. $\tbarbound^t$
is $\barbound^t$ after thresholding it to keep it in the valid range $[0,W]^{2|E|}$.
Clearly,
\begin{align}
\baralf^t_P\preceq \tbarbound^t(+)\preceq \barbound^t(+)\, .
\label{eq:path_boundthresh_sandwich}
\end{align}
In words, the right inequality says that thresholding can only reduce the
positively bounding components of $\barbound^t$ and only increase the
negatively bounding components.
Let $\barbet^t=(\baralf^t_{\cG- E_P}, \tbarbound^t)$.
Note that $\barbet^t$ is a valid message vector.
It follows from (\ref{eq:bounding_range}) and (\ref{eq:path_boundthresh_sandwich})
that $\Bo(q,\delta)$ (cf. definition \ref{def:Bo_overloaded}) holds for $\barbet^t$.
We let $\barbet^t$ evolve for
one time step under the original dynamics.
Let the resulting messages on the path be $\bartightbound^{t+1}(+)$.
By Lemma \ref{lemma:monotonicity} and the first inequality in (\ref{eq:path_boundthresh_sandwich}),
 we know that $\bartightbound^{t+1}(+) \succeq \baralf^{t+1}_P$.
We show that $\barbound^{t+1}(+) \succeq \bartightbound^{t+1} (+)$.

Consider a node $i$ on $P$ with $i$ even, $0<i<\len$. We show that $\bound{i}{i+1}^{t+1} \geq
\tightbound{i}{i+1}^{t+1}$. Firstly, we know that $\soff{i-1}{i}^t(+) \geq \off{i-1}{i}^* + (\Delta-\delta)$
from the first inequality in (\ref{eq:bounding_range}).
We have,
\begin{align*}
\soff{i-1}{i}^t(+)&\stackrel{(a)}{=}w_{(i-1,i)}-\bound{i-1}{i}^t-\left(\frac{w_{(i-1,i)}-\bound{i-1}{i}^t-\bound{i}{i-1}^t}{2}\right)_+\\
&\stackrel{(b)}{=}\left(w_{(i-1,i)}-\bound{i-1}{i}^t\right)_+-\left(\frac{w_{(i-1,i)}-\bound{i-1}{i}^t-\bound{i}{i-1}^t}{2}\right)_+\\
&\stackrel{(c)}{\geq} \left(w_{(i-1,i)}-\tbound{i-1}{i}^t\right)_+-\left(\frac{w_{(i-1,i)}-\tbound{i-1}{i}^t-\tbound{i}{i-1}^t}{2}\right)_+\\
&\stackrel{(d)}{=} \off{i-1}{i}|_{\bartightbound^t}
\end{align*}
(a) follows from Lemma
\ref{lemma:CorrectMatching}. (b) follows from $ \soff{i-1}{i}^t(+) > 0 \Rightarrow w_{(i-1,i)}-\bound{i-1}{i}^t > 0$.
(c) holds since $\tbound{i-1}{i}^t \geq \bound{i-1}{i}^t, \tbound{i}{i-1}^t \leq \bound{i}{i-1}^t$.
(d) follows since $\bound{i}{i-1}^t \geq \alf{i}{i-1}^t \geq 0$.
From Lemma \ref{lemma:outsideOff_ub} for all $j \in \partial i \backslash i-1$ we have,
\begin{align*}
\off{j}{i}^t|_{\bartightbound^t} \leq
\gamma_i^*-\sigma_q-\delta+\Delta \leq \off{i-1}{i}^* + \Delta - \delta \leq \soff{i-1}{i}^t(+).
\end{align*}
Thus, $\soff{i-1}{i}^t(+) \geq \max_{j \in \partial i\backslash i+1} \off{j}{i}^t|_{\bartightbound^t} $. Also,
 $\bound{i}{i+1}^{t} \geq \tbound{i}{i+1}^{t}=\bet{i}{i+1}^{t}$.
Hence, $\bound{i}{i+1}^{t+1} \geq \tightbound{i}{i+1}^{t+1}$.

Now consider a node $i$ on $P$ with $i$ odd, $0<i<\len$. We show that $\bound{i}{i+1}^{t+1} \leq
\tightbound{i}{i+1}^{t+1}$. Under the simplified dynamics, in light of Lemma
\ref{lemma:CorrectMatching}, we know that
\begin{align*}
\soff{i-1}{i}^t&=w_{(i,i-1)}-\bound{i-1}{i}^t-\left(\frac{w_{(i,i-1)}-\bound{i-1}{i}^t-\bound{i}{i-1}^t}{2}\right)_+\\
&\leq (w_{(i,i-1)}-\tbound{i-1}{i}^t)_+-\left(\frac{w_{(i,i-1)}-\tbound{i-1}{i}^t-\tbound{i}{i-1}^t}{2}\right)_+\\
&= \off{i-1}{i}|_{\bartightbound^t(+)}\\
&\leq \max_{j \in \partial i\backslash i+1} \off{j}{i}^t|_{\bartightbound^t}
\end{align*}
where we used $\tbound{i-1}{i}^t \leq \bound{i-1}{i}^t$, $\tbound{i}{i-1}^t \geq \bound{i}{i-1}^t$.
Also,
 $\bound{i}{i+1}^{t} \leq \tbound{i}{i+1}^{t}=\bet{i}{i+1}^{t}$.
Hence, $\bound{i}{i+1}^{t+1} \leq \tightbound{i}{i+1}^{t+1}$.

The required inequalities for $\bound{i}{i-1}^{t+1}$, can be proved similarly.

This leaves us with checking that $\tightbound{0}{1}^{t+1} \leq \bound{0}{1}^{t+1}$
and $(-1)^\len \bound{\len}{\len-1}^{t+1} \geq (-1)^\len \tightbound{\len}{\len-1}^{t+1}$.
We prove the second inequality and the first one follows from repeating the analysis for the $\len$ even case.

If $\len \notin V(\cC_q) \Rightarrow \len \in \cG_q$ and the result follows from Lemma \ref{lemma:CondStrucPersistence}
applied to the initial state $\barbet^t$ for one step along with (\ref{eq:bounding_range}).
Else, $\len \in V(\cC_q) \Rightarrow (\len,\len-1) \in M^*, \gamma_\len^*=\sigma_q, \alf{\len}{\len-1}^*=0$.

Case (i): $\len$ odd\\
Follows from $\bound{\len}{\len-1}^{t+1} \leq 0 \leq \tightbound{\len}{\len-1}^{t+1}$.

Case (ii): $\len$ even\\
From Lemma \ref{lemma:outsideOff_ub} we  obtain
$\off{j}{\len}|_{\bartightbound^t} \leq (\Delta -\delta) \quad \forall\, j \in \partial \len \backslash \len-1$.
Also, $\bet{\len}{\len-1}^t = \tbound{\len}{\len-1}^t \leq \bound{\len}{\len-1}^t$.
$\tightbound{\len}{\len-1}^{t+1} \leq \bound{\len}{\len-1}^{t+1}$ follows.

$\barbound^{t+1}(-) \preceq \baralf^t_P$ follows similarly. Induction on $t$ completes the proof.
\end{proof}
We also generalize our analysis to the case  for use in analysis the cases of blossom and bicycle.
Suppose $P$ is a subgraph of $(V_{\textup{ext}}(\cC_q), E(\cC_q))$, with the subset $V_P$ of nodes in
$V_{\textup{ext}}$ appropriately given the aliases $0,1,\ldots, \len$. We define the bounding
processes as in Section \ref{subsubsec:path_descent} on P, using the matching $M_P^*$ that is the restriction of $M^*$ to $P$. We prove the following refinement of Lemma \ref{lemma:BoundingBounds}.

\begin{lemma}\label{lemma:BoundingBounds_refined}
Let $G$ be an instance admitting a unique NB solution with
gap $\sigma$, and assume its KT sequence
to be $(\cC_0, \cC_1,\cC_2,\dots,\cC_k)$. Let
path $P=P_\len$ be a subgraph of $(V_\textup{ext}(\cC_q), E(\cC_q))$
for some  $q \in \{1,\ldots, k\}$, such that
$\partial P= \{(i,j): i \in V_\textup{ext}(\cC_q)
\backslash V(P), j \in V(P), (i,j) \in E(\cC_q)\}$ is identical
to $\partial P_\textup{end}= \{(i,j): i \in V_\textup{ext}(\cC_q)
\backslash V(P), j \in \{0,\len\}, (i,j) \in E(\cC_q)\}$ i.e.
$P$ touches the rest of $\cC_q$ only at its ends.
Assume $\Bo_\Delta(q,\delta,0)$ holds. Also,
suppose that
\begin{align*}
U_{G,\partial P}(\baralf^t) \leq \Delta -\delta' \qquad  \forall t \geq 0
\end{align*}
for $\delta' \in [0, \delta]$.
If we denote by $\{\barbound^t(s)\}_{t\ge 0}$
the $(s,\Delta,\delta')$ bounding process on $P$, then for any
$t\ge 0$:
\begin{eqnarray}
\barbound^t(-)\preceq \baralf^t_P\preceq \barbound^t(+)\, .
\label{eq:real_path_bound_refined}
\end{eqnarray}
\end{lemma}

\begin{proof}
$\Bo_\Delta(q,\delta,0) \Rightarrow \Bo_\Delta(q,\delta',0)$.
The proof here is almost identical to the one for Lemma \ref{lemma:BoundingBounds},
with $\delta'$ replacing $\delta$.
We use induction on $t$. The required inequalities at $t+1$ for internal messages
are proved exactly as before. This leaves us with the boundary messages.

We must show $\bound{0}{1}^{t+1}(-) \leq \alf{0}{1}^{t+1} \leq \bound{0}{1}^{t+1}(+)$.
$\bound{0}{1}^{t}(-) \leq \alf{0}{1}^{t} \leq \bound{0}{1}^{t}(+)$ holds by hypothesis.
Thus it suffices to show that
\begin{align}
|\max_{j \in \partial 0 \backslash 1} \off{j}{0}^t - \alf{0}{1}^*| \leq \Delta -\delta'
\label{eq:bestoffer_to0_bound}
\end{align}
If $0 \notin V(\cC_q) \Rightarrow 0 \in \cG_q$ and the result follows from Lemma \ref{lemma:CondStrucPersistence}
applied to the initial state $\baralf^t$ for one step along with (\ref{eq:bounding_range}).
Else, $0 \in V(\cC_q) \Rightarrow 0 \in M^*$. Two cases arise.

Case (i): $\alf{0}{1}^* > 0$\\
In this case, there exists a $(v,0) \in \partial P$ for some $v \in V_{\textup{ext}}(\cC_q) \backslash P$
s.t. $\off{v}{0}^*=\alf{0}{1}^* \geq \gamma_0^*-\sigma_q$.\\
$\off{v}{0}^t \geq \off{v}{0}^* - (\Delta -\delta')$
yields the lower bound in (\ref{eq:bestoffer_to0_bound}).\\
For any $(j,0) \in \partial P$, we have
$\off{j}{0}^t \leq \off{j}{0}^*+ \Delta -\delta' \leq \alf{0}{1}^*+ \Delta -\delta'$.
For any $(j,0) \notin \partial P \cup \{(1,0)\}$, Lemma \ref{lemma:outsideOff_ub} yields
$\off{j}{0}^t \leq (\gamma_0^*-\sigma_q-\sigma +\Delta)_+ \leq \alf{0}{1}^*+ \Delta -\delta'$.
Thus we have the required upper bound.

Case (ii): $\alf{0}{1}^* = 0$\\
In this case there is no $(v, 0) \in \partial P$, and we know that
$(0,1) \in M^*$ and $\gamma_0^*=\sigma_q$.
The lower bound is trivial in this case.
Lemma \ref{lemma:outsideOff_ub} leads to
$\off{j}{0}^t \leq (\Delta - \sigma)_+ \,, \forall \, j \in \partial 0 \backslash 1$.
The upper bound follows. The inequality
\[
(-1)^\len\bound{\len}{\len-1}^{t+1}(-) \leq (-1)^\len\alf{\len}{\len-1}^{t+1} \leq (-1)^\len\bound{\len}{\len-1}^{t+1}(+)
\]
follows similarly.
\end{proof}

%
%
%
%
\subsection{Blossom convergence}
\label{app:Blossom}

The blossom consists of a `stem' and a `cycle'. There is a node
shared by the stem and the cycle with its matching partner in the
stem; call this node $0$. Starting at $0$ and going down the stem to its anchored
end, we label the next node $1_\s$, the one after that $2_\s$ and so on up to $(\len_\s)_\s$,
where $\len_\s$ is the number
of edges in the stem. We denote by $P_\s$ the path $0-1_\s-2_\s- \ldots -(\len_\s)_\s$.
The matching of interest in $P_\s$ is always $M_{P_\s}^*=((0,1_\s), (2_\s,3_\s), \ldots )$
with every alternate edge matched.
Similarly, starting at $0$ and going around the cycle in
an arbitrary direction, we label the next node $1_{\c}$, the one after that $2_\c$,
and so on up to $(\len_\c-1)_\c$ after which we return to $0$.
Here $\len_\c$ is the (odd) number of edges in the cycle.
The cycle can be thought of as a path that is closed on itself i.e.
the two ends of the path are the same. Denote by $P_{\c}$ this path
$0_\c - 1_{\c} - 2_\c- \ldots -(\len_\c-1)_\c-(\len_\c)_c$, where $0_\c$ and $\len_\c$
 are copies of the node $0$. The matching of interest in
 $P_\c$ is $M_{P_\c}^*=\{(i_\c,(i+1)_\c): 1 \leq i \leq (\len_\c-2), i \mbox{ odd}\}$.
 Note that both end edges $(0_\c,1_{\c})$ and $((\len_\c-1)_\c,(\len_\c)_c)$ i.e. the edges
 $(0,1_{\c})$ and $((\len_\c-1)_\c,0)$ in the blossom are unmatched.

Let $G$ be an instance admitting a unique NB solution
$\baralf^*$ with
gap $\sigma$. Suppose $\cC_q$ corresponds to a blossom i.e.
$(V_{\textup{ext}}(\cC_q), E(\cC_q))$ is a blossom with stem $P_\s$ and
cycle mapped to $P_\c$. Nodes in
$V_{\textup{ext}}$ are appropriately given the aliases
$0, 1_\s, \ldots, (\len_\s)_\s, 1_\c, 2_\c, \ldots, (\len_\c)_\c$.
$0_\c$ and $\len_\c$ are copies of node $0$.

We develop a new bounding process for the path $P_\c$ and show that
this bounding process `contracts' reducing the error on $P_\c$.
The fact that the cycle is odd plays a critical
role in this contraction. Next, the analysis
from Section \ref{subsubsec:path_descent} is used to show reduction of the error in $P_\s$.
We iterate the reduction in error in the stem and cycle to show that the overall
error in the blossom reduces.

In the rest of this section, we sometimes drop the subscript $\c$ eg. Node $1$ refers to node $1_\c$, since
the new results developed are for $P_\c$ and not $P_\s$.

Let $M_{P_\c}^*$ be the restriction of $M^*$ to $P_\c$.
We know that $\len=\len_\c$ is odd.
Assume $\Delta>0$, $\delta\le \min(\sigma,\Delta)$, $s\in \{+1,-1\}$
and $\Times=(\Times_+,\Times_-)$ be a partition of $\naturals \cup \{0\}$
(i.e. $\Times_+\cup\Times_- = \naturals \cup \{0\}$, $\Times_+\cap\Times_- = \emptyset$).
Let $\delta_\c, \delta_\s \in (0, \delta]$.
The $(s,\Delta,\delta,\Times,\delta_\c,\delta_\s)$-bounding process on $P_\c$ is
the sequence of message assignments $\{\barbound^t\}_{t\ge 0}$ produced by the simplified
dynamics on $P_\c$ wrt $M_{P_\c}^*$, and
the boundary conditions
\begin{eqnarray}
b_\L^t = \bound{0}{1}^*(s)+s\, \delta_\s\, ,\qquad
b_\R^t = \bound{\len}{\len-1}^*(s)\, , \qquad \mbox{if $t\in\Times_+$,} \label{eq:boundary_tplus}\\
b_\L^t =\bound{0}{1}^*(s)\, ,\qquad
b_\R^t = \bound{\len}{\len-1}^*(s)-s\delta_\s\, , \qquad \mbox{if $t\in\Times_-$.}
\label{eq:boundary_tminus}
\end{eqnarray}
with the initial condition
(for $i\in\{0,\dots,\len-1\}$)
\begin{eqnarray}
\bound{i}{i+1}^0 = \bound{i}{i+1}^*(s)+s(-1)^i\delta_\c\, ,\;\;\;\;\;\;\;
\bound{i+1}{i}^0 = \bound{i+1}{i}^*(s)-s(-1)^{i}\delta_\c\, ,
\end{eqnarray}
where $\barbound^*(s)$ is given by Eq.~(\ref{eq:bound_fp}).

The fact that only
one of the boundary conditions is different from $\barbound^*$ at each time is instrumental
in the `contraction' of this bounding process. The odd length of the cycle is crucial
in ensuring that such a bounding process, with a `helpful' boundary condition at atleast
one end, sandwiches the actual dynamics (cf. Lemma \ref{lemma:BoundingBounds_pc}).

\begin{lemma}
Let  $\{\barbound^t\}_{t\ge 0}$ be the $(s,\Delta,\delta,\Times,\delta_\c,\delta_\s)$-bounding process
on $P_\c$ wrt $M_{P_\c}^*$ for some $\delta \leq \min(\sigma,\Delta)$ and $\delta_\c, \delta_\s \in (0, \delta]$.
Then for arbitrary $\Times=(\Times_+,\Times_-)$, $\exists c_1$ finite and $c_2 >0$ s.t. for any
$\eps >0$,

\begin{align}
||\barbound^t-\barbound^*||_\infty &\leq \delta_\s + \eps\\
\max(|\bound{1}{0}^t - \bound{1}{0}^*|, |\bound{\len-1}{\len}^t-\bound{\len-1}{\len}^*|) &\leq \delta_\s \left(1-\frac{c_2}{\len^{3}}\right) +\eps
\end{align}
hold $\forall \, t \geq c_1 \len^{2.1}\lceil\log(1+\delta_c/\epsilon)\rceil$.
\label{lemma:blossom_cycle_bounding_descent}
\end{lemma}

\begin{proof}
After taking differences with respect to
$\barbound^*$, $\sbarDelta^t=\barbound^t-\barbound^*$,  the initialization becomes
\begin{eqnarray}
\sDel{i}{i+1}^0 = s(-1)^i\delta_\c\, ,\;\;\;\;\;\;\;
\sDel{i+1}{i}^0 = -s(-1)^{i}\delta_\c\, ,\qquad \mbox{for $i\in\{0,\dots,\len-1\}$}
\end{eqnarray}
the boundary conditions
\begin{eqnarray}
\hat{b}_\L^t = s\, \delta_\s\, ,\;\;\;\;\;\;\;
\hat{b}_\R^t = 0\, ,&& \mbox{if $t\in\Times_+$,}\\
\hat{b}_\L^t =0\, ,\;\;\;\;\;\;\;
\hat{b}_\R^t = -s\delta_\s\, ,&& \mbox{if $t\in\Times_-$.}
\end{eqnarray}
and the dynamics reduces to
(\ref{eq:matching_simple_path}), (\ref{eq:nonmatch_simple_path}) for internal edges,
and
\begin{align}
\sDel{0}{1}^{t+1} &= \damp \hat{b}_\L^t +(1-\damp)\sDel{0}{1}^t\\
\sDel{\len}{\len - 1}^{t+1} &= \damp \hat{b}_\R^t +(1-\damp)\sDel{\len}{\len-1}^t
\end{align}
 for boundary edges, for $t \geq 0$.

Thus, in order to prove Lemma \ref{lemma:blossom_cycle_bounding_descent}, it suffices
to prove  the Lemma \ref{lemma:Pc_RW_bound} which is then combined with the contraction result
in Lemma \ref{lemma:simplified_convergence}.
\end{proof}

\begin{lemma}
Consider a path $P_\len$ (with no edge weights) for $\len$ odd, a valid path matching $M_P$
and arbitrary partition
$\Times=(\Times_+,\Times_-)$ of $\mathbb{N}\cup\{0\}$. Consider the dynamics on the vector
$\barmass$ defined by
(\ref{eq:matching_simple_path_bound}) and (\ref{eq:nonmatch_simple_path_bound})
for interior messages and
\begin{align}
\mass{0}{1}^{t+1}&=\damp \tilde{b}_\L^t +(1-\damp)\mass{0}{1}^t\, ,\\
\mass{\len}{\len-1}^{t+1}&=\damp \tilde{b}_\R^t +(1-\damp)\mass{\len}{\len-1}^t\, ,
\end{align}
for boundary messages, with
\begin{align}
\tilde{b}_\L^t=1\,,\qquad \tilde{b}_\R^t=0 \qquad \mbox{if $t\in\Times_+$,}\\
\tilde{b}_\L^t=0\,,\qquad \tilde{b}_\R^t=1 \qquad \mbox{if $t\in\Times_-$.}
\end{align}
Assume the initial condition to be
$\barmass^0=\underline{0}$. Then for all $t \geq 0$ we have
\begin{align}
\mass{i}{j}^t &\leq 1 \qquad \forall\,  (i,j) \in P\, ,\label{eq:Bound1RW}\\
\max(\mass{1}{0}^t, \mass{\len-1}{\len}^t) &\leq 1-\frac{c_2}{\len^{3}}\, .
\label{eq:Bound2RW}
\end{align}
for some $c_2 > 0$ not dependent on $\len$.
\label{lemma:Pc_RW_bound}
\end{lemma}
\begin{proof}
Equation (\ref{eq:Bound1RW}) follows immediately by induction
using equations (\ref{eq:matching_simple_path_bound}) and
(\ref{eq:nonmatch_simple_path_bound}).

In order to prove Eq.~(\ref{eq:Bound2RW}),
recall the random walk interpretation of the  equations
(\ref{eq:matching_simple_path_bound}) and (\ref{eq:nonmatch_simple_path_bound})
provided in the proof of Lemma \ref{lemma:RW}.
With the new boundary conditions, $\mass{i}{j}^t$ is the expected
number of particles on edge $(i,j)$, and directed from $i$ to $j$,
when the path system is initially empty (since $\barmass^0=\underline{0}$)
and a particle is injected with probability $\damp$ at times $t\in\Times_+$
on edge $(0,1)$  (directed to the right), and at times $t\in\Times_+$
on edge $(\len,\len-1)$ (directed to the left).

By invariance of the random walk process under
time traslation, we can consider a process that is initiated at time $-t$,
with particles injected according a partition $\Times^{(-t)} =
(\Times^{(-t)}_+,\Times^{(-t)}_-)$
of $\{0,-1,-2,\dots,-t\}$. In order to prove
Eq.~(\ref{eq:Bound2RW}) it is sufficient to bound the mass on boundary edges
at time $0$, $\max(\mass{1}{0}^0, \mass{\len-1}{\len}^0)$.
By monotonicity of this quantity with respect to the particle injection,
we can consider an extension $\Times^{(-\infty)} =
(\Times^{(-\infty)}_+,\Times^{(-\infty)}_-)$ of the
above partition, and focus on the case $t\to\infty$, which we will do
hereafter.

Since the random walk process is left unchanged under left-right
inversion of the path ($\len$ is odd), we have $\mass{\len-1}{\len}^0 =
\imass{1}{0}^0$, where $\imass{i}{j}^t$
is the expected number of particles on edge $(i,j)$
under boundary condition that inverts
$\Times^{(-\infty)}_+$ and  $\Times^{(-\infty)}_-$.
It is therefore sufficient to bound
\begin{eqnarray}
\max(\mass{1}{0}^0, \imass{1}{0}^0) = \mass{1}{0}^0+
\imass{1}{0}^0-\min(\mass{1}{0}^0, \imass{1}{0}^0)\, .
\end{eqnarray}
Notice that, by linearity,  $\mass{1}{0}^0+
\imass{1}{0}^0$ is the expected number of particles on edge
$(0,1)$ (directed towards $0$) when a particle is injected
at each time $-1$, $-2$, $\dots$ at both ends with probability
$\damp$. It is easy to check that the stationary distribution
is, in this case, of $1$ expected particle per edge, yielding
$\mass{1}{0}^0+\imass{1}{0}^0=1$.

We are left with the task of showing that
$\min(\mass{1}{0}^0, \imass{1}{0}^0)\ge c_2/\len^{3}$. Again by monotonicity
with respect to the particle injection, it is sufficient to
show that there exist at least one time $t<0$
such that the following two facts are true:
$(1)$ A particle injected at time $t$ on edge $(0,1)$
is on edge $(1,0)$ at time $0$ with probability larger than $ c_2/\len^{3}$;
$(2)$ A particle injected at time $t$ on edge $(\len,\len-1)$
is on edge $(1,0)$ at time $0$ with probability larger than $ c_2/\len^{3}$.

Recall that the process when watched on
matching edges, is a lazy random walk.
Therefore, we reduced our problem to proving the following elementary
fact about the gambler's ruin model: There exists a time
$t$ such that both probabilities that the gambler has one dollar
and that he has $m'-1$ dollars ($m'$ being the maximum)
are at least $c_2'/(m')^{3}$. It is easy to check that this is the
case, for some $t = \Theta\left((m')^2\right)$ (see, for example, \cite{Spitzer}).
Note that this analysis, which
holds for $\len \geq \len_0$ for some $\len_0<\infty$,
suffices -- for any fixed value of $\len$, we know that $\exists t < 0$
such that a particle injected at time $t$ on edge $(0,1)$
is on edge $(1,0)$ with non-zero probability AND
a particle injected at time $t$ on edge $(\len,\len - 1)$
is on edge $(1,0)$ with non-zero probability.
\end{proof}

\begin{lemma}\label{lemma:CorrectMatching_pc}
Let  $\{\barbound^t\}_{t\ge 0}$ be the $(s,\Delta,\delta,\Times,\delta_\c,\delta_\s)$-bounding process
on $P_\c$ wrt $M_{P_\c}^*$ for some $\delta \leq \min(\sigma,\Delta)$ and $\delta_\c, \delta_\s \in (0, \delta]$.
Then for any $t$, and for any $(i,j)\in E_P$
\begin{eqnarray}
\bound{i}{j}^t +\bound{j}{i}^t-w_{ij}\le 0&&
\mbox{ if $(i,j)\in M_{P_\c}^*$,}\\
\bound{i}{j}^t +\bound{j}{i}^t-w_{ij}\ge 0&&
\mbox{ if $(i,j)\not\in M_{P_\c}^*$.}
\end{eqnarray}
\end{lemma}
\begin{proof}
Consider $s=(+)$. We will use Lemma \ref{lemma:monotonicity_simplified} to bound $\barbound^t$
from both sides. Let $\{b_\L^t\}, \{b_\R^t\}$ denote the boundary conditions for
$\{\barbound^t(+)\}$ as defined in (\ref{eq:boundary_tplus}), (\ref{eq:boundary_tminus}).
Let $b_{*,\L}=\bound{0}{1}^*\,,\ b_{*,\R}=\bound{\len}{\len-1}^*$.
Let $b_{0,\L}=\bound{0}{1}^*+\delta_\c\,,\ b_{0,\R}=\bound{\len}{\len-1}^*-\delta_\c$.
(Recall that $\len$ is odd.)
Clearly, $\barbound^0$ is a fixed point with respect to the constant boundary conditions
$b_{0,\L}, b_{0,\R}$. We have,
\begin{align*}
b_{*,\L} \leq b_\L^t \leq b_{0,\L} \qquad b_{*,\R} \geq b_\R^t \geq b_{0,\R} \qquad \forall \, t \geq 0
\end{align*}
Also, $\barbound^* \preceq \barbound^0$.
Using Lemma \ref{lemma:monotonicity_simplified} twice we have
\begin{align}
\barbound^*(+) \preceq \barbound^t(+) \preceq \barbound^0(+) \qquad \forall \,t \geq 0
\label{eq:bounding_range_pc}
\end{align}
Recall $\delta_\c \leq \delta$.
The rest of the proof mirrors that of Lemma \ref{lemma:CorrectMatching}.

A similar proof can be constructed for $s=(-)$.
\end{proof}

\begin{lemma}\label{lemma:BoundingBounds_pc}
Let $G$ be an instance admitting a unique NB solution with
gap $\sigma$, and assume its KT sequence
to be $(\cC_0, \cC_1,\cC_2,\dots,\cC_k)$, with $\cC_q, q \in \{1,\ldots, k\}$ a blossom
mapped to stem $P_\s$ and cycle $P_\c$. Suppose $\Bo_\Delta(q,\delta,0)$ holds
for some $\delta \leq \min(\sigma, \Delta)$.
Also, assume
\begin{align}
U_{G,P_\s}(\baralf^t) &\leq \Delta-\delta+\delta_\s  \qquad \forall \,t \geq 0 \\
U_{G,P_\c}(\baralf^0) &\leq \Delta-\delta+\delta_\c
\end{align}
for $\delta_\c, \delta_\s \in (0,\delta]$.
Define
\begin{align*}
\Times_+ &=\{t: t\in \mathbb{N}\cup\{0\}, \off{1_\s}{0}^t \geq \off{1_\s}{0}^*\}\\
\Times_- &=\{t: t\in \mathbb{N}\cup\{0\}, \off{1_\s}{0}^t < \off{1_\s}{0}^*\}\\
\Times &=(\Times_+,\Times_-)
\end{align*}
Denote by $\{\barbound^t(s)\}_{t\ge 0}$
the $(s,\Delta,\delta,\Times,\delta_c,\delta_s)$ bounding process on $P_\c$. Then for any
$t\ge 0$:
\begin{eqnarray}
\barbound^t(-)\preceq \baralf^t_{P_\c}\preceq \barbound^t(+)\, .
\label{eq:real_pc_bound}
\end{eqnarray}
\end{lemma}

\begin{proof}
The proof is by induction on $t$. (\ref{eq:real_pc_bound}) holds at
$t=0$ by assumptions. Suppose it holds at $t$. The proof of the required inequalities
at $t+1$ for internal messages mirrors the proof of Lemma \ref{lemma:BoundingBounds}.

We must show $\bound{0}{1}^{t+1}(-) \leq \alf{0}{1}^{t+1} \leq \bound{0}{1}^{t+1}(+)$.
By hypothesis, $\bound{0}{1}^{t}(-) \leq \alf{0}{1}^{t} \leq \bound{0}{1}^{t}(+)$.
From Claim \ref{claim:offer_non_expansion}, we know that
\begin{align}
|\off{1_\s}{0}^t-\off{1_\s}{0}^*| \leq \Delta - \delta +\delta_s
\label{eq:off_1s_0}
\end{align}
Case (i): $\off{1_\s}{0}^t \geq \off{1_\s}{0}^*$ i.e. $t \in \Times_+$\\
It suffices to show that
\begin{align}
\max_{j \in \partial 0 \backslash 1} \off{j}{0} \in
[\alf{0}{1}^*-(\Delta-\delta), \alf{0}{1}^*+(\Delta-\delta)+\delta_s]
\end{align}
The lower bound follows from
$\off{1_\s}{0}^t \geq \off{1_\s}{0}^*=\alf{0}{1}^* \geq \alf{0}{1}^*-(\Delta-\delta)$.
We have
\begin{align}
\off{j}{0} \leq \off{j}{0}^* + \Delta \leq \off{1_\s}{0}^*-\sigma_q+\Delta
\leq \alf{0}{1}^*+(\Delta-\delta) \,, \qquad \forall \, j \in \partial 0 \backslash 1_\s
\label{eq:off_not1s_0}
\end{align}
(\ref{eq:off_1s_0}) and (\ref{eq:off_not1s_0}) give the upper bound.

Case (ii): $\off{1_\s}{0}^t < \off{1_\s}{0}^*$ i.e. $t \in \Times_-$\\
It suffices to show that
\begin{align}
\max_{j \in \partial 0 \backslash 1} \off{j}{0} \in
[\alf{0}{1}^*-(\Delta-\delta)-\delta_s, \alf{0}{1}^*+(\Delta-\delta)]
\end{align}
The lower bound follows from (\ref{eq:off_1s_0}). The upper bound follows from
(\ref{eq:off_not1s_0}) and $\off{1_\s}{0}^t < \off{1_\s}{0}^*=\alf{0}{1}^*$.

$\bound{\len}{\len-1}^{t+1}(-) \geq \alf{\len}{\len-1}^{t+1} \geq \bound{\len}{\len-1}^{t+1}(+)$ follows similarly.

Induction on $t$ completes the proof.
\end{proof}

\begin{lemma}
Let $G$ be an instance admitting a unique NB solution
$\baralf^*$. Suppose $\cC_q$ corresponds to a blossom and
suppose the $\Bo(q, \delta, 0)$ holds for some $\delta \leq \min(\sigma, \Delta)$.
We have
\begin{align}
U_{G,P_\s}(\baralf^t)
&\leq \Delta-\delta + \frac{\delta}{10n}
\qquad \forall \, t \geq C n^6
\label{eq:Ps_bound}
\end{align}
for some $C$ finite.
\label{lemma:blossom_stem_bound}
\end{lemma}
\begin{proof}
Define $\delta_{\s}(0)=\delta_{\b}(0)=\delta$. For $i=1, 2, \ldots, N$ define
\begin{align}
\delta_\b(i) &= \delta \left(1-\frac{c_2}{4n^{3}}\right)^{i-1}\left(1-\frac{c_2}{2n^{3}}\right)\\
\delta_\s(i) &= \delta \left(1-\frac{c_2}{4n^{3}}\right)^{i}
\end{align}
where $N$ is such that $\delta_\s(N) < \frac{\delta}{10n} \leq \delta_\s(N-1)$.
Clearly, $N \leq c_5 n^{3.1}$ for some finite $c_5$.

We show that $\exists c_3, c_4$ finite such that $\forall i \in \{1, 2, \ldots, N\}$
\begin{align}
\max(|\alf{1_\c}{0}^t-\alf{1_\c}{0}^*|, |\alf{(\len_\c-1)_\c}{0}^t-\alf{(\len_\c-1)_\c}{0}^*|)
&\leq \Delta-\delta + \delta_\b(i)
\qquad \forall \, t \geq \left( ic_3+(i-1)c_4\right) n^{2.2}
\label{eq:Pc_iterative_bound}\\
U_{G,P_\s}(\baralf^t)
&\leq \Delta-\delta + \delta_\s(i)
\qquad \forall \, t \geq i\left( c_3+c_4\right) n^{2.2}
\label{eq:Ps_iterative_bound}
\end{align}
where $c_2$ if from Lemma \ref{lemma:blossom_cycle_bounding_descent}.

We prove (\ref{eq:Pc_iterative_bound}), (\ref{eq:Ps_iterative_bound}) by induction on the iteration number $i$,
defining $c_3$ and $c_4$ along the way.

Clearly, (\ref{eq:Ps_iterative_bound}) holds for $i=0$. Suppose it holds for $i$. The iteration $i+1$ consists of
two phases.

Phase I:\\
Use Lemma \ref{lemma:BoundingBounds_pc} and Lemma \ref{lemma:blossom_cycle_bounding_descent}
with $\delta_s=\delta_\s(i), \delta_\c=\delta, \epsilon = \frac{\delta}{10n}.\frac{c_2}{2n^{3}}$.
We have
\begin{align}
\max(|\alf{1_\c}{0}^t-\alf{1_\c}{0}^*|, |\alf{(\len_\c-1)_\c}{0}^t-\alf{(\len_\c-1)_\c}{0}^*|)
&\leq \Delta-\delta + \delta_\b(i)( 1- \frac{c_2}{n^{3}}) + \frac{\delta}{10n}.\frac{c_2}{2n^{3}}\\
&\leq \Delta-\delta + \delta_\b(i)( 1- \frac{c_2}{2n^{3}})
\end{align}
holds after additional time
$c_3 n ^{2.2}$, where $c_3$ does not depend on $i$.
Hence, (\ref{eq:Pc_iterative_bound}) holds for $i+1$.

Phase II:\\
Use Lemma \ref{lemma:BoundingBounds_refined} on $P_\s$  with
$\delta'= \delta - \delta_\s(i+1),$ and Lemma \ref{lemma:BoundConv} with $\epsilon = \frac{\delta}{10n}.\frac{c_2}{4n^{3}}$
on the $(s, \Delta, \delta')$-bounding process
to show that
(\ref{eq:Ps_iterative_bound}) holds for $i+1$. The constant $c_4$ is chosen such that
$c_4 n^{2.2} \geq c_\RW n^{2.1}\left(1+ \log(\delta/\eps)\right) \geq c_\RW n^{2.1}\left(1+ \log(\delta'/\eps)\right)$,
since $\delta' \leq \delta$.
\end{proof}

\begin{proof}[Proof (Lemma \ref{lemma:KeyLemma}: Blossom)]
We use Lemma \ref{lemma:blossom_stem_bound} to bound the error on $P_\s$. Next we use
Lemma \ref{lemma:BoundingBounds_pc} and Lemma \ref{lemma:blossom_cycle_bounding_descent}
with $\delta_\s=\delta_\s(N), \delta_\c=\delta, \epsilon = \frac{\delta}{20n}$ to show
that
\begin{align}
U_{G,P_\c}(\baralf^t) \leq \Delta - \delta + \frac{3\delta}{20n}
\end{align}
after additional time $c_5 n^{2.2}$ for finite $c_5$. Finally, we use Lemma \ref{lemma:outgoing_error_bound} with
$\eps = \frac{\delta}{20n}$. This finally gives that $\Bo(q+1, \delta(1-(5n)^{-1}),t)$ holds for all
$t \geq Cn^6$ as required.
\end{proof}
%
%
\subsection{Bicycle convergence}
\label{app:Bicycle}

An almost identical argument works for the  bicycle as for the blossom.
We define $P_\s$ as the `rod' or `frame' of the bicycle, joining the
two cycles mapped to $P_{\c 1}$ and $P_{\c 2}$, analogously to the definitions
of $P_\s$ and $P_\c$ for the blossom. Lemma \ref{lemma:blossom_stem_bound}
can be shown to hold for the bicycle as well, with an almost identical procedure.
In this case also, we have an iterative descent, alternating between the `rod' and the
two cycles (which descend simultaneously), as per
\begin{align}
U_{G,\partial P_\s}(\baralf^t)
&\leq \Delta-\delta + \delta_\b(i)
\qquad \forall \, t \geq \left( ic_3+(i-1)c_4\right) n^{2.2}
\label{eq:Pc_iterative_bound_bicycle}\\
U_{G, P_\s}(\baralf^t)
&\leq \Delta-\delta + \delta_\s(i)
\qquad \forall \, t \geq i\left( c_3+c_4\right) n^{2.2}
\label{eq:Ps_iterative_bound_bicycle}
\end{align}
where $\partial P= \{(i,j): i \in V_\textup{ext}(\cC_q)
\backslash V(P), j \in V(P), (i,j) \in E(\cC_q)\}$.
Phase I works by using Lemmas \ref{lemma:BoundingBounds_pc} and Lemma \ref{lemma:blossom_cycle_bounding_descent}
\textit{simultaneously for both $P_{\c 1}$ and $P_{\c 2}$}. Phase II works exactly as before.

\begin{proof}[Lemma \ref{lemma:KeyLemma}: Bicycle]
We first use Lemma \ref{lemma:blossom_stem_bound}. Next, we use
Lemma \ref{lemma:BoundingBounds_pc} and Lemma \ref{lemma:blossom_cycle_bounding_descent} simultaneously
for each of $P_{\c 1}$ and $P_{\c 2}$. Finally, we use Lemma \ref{lemma:outgoing_error_bound}.
Note that we thus prove an identical bound on the descent time of a bicycle as for a blossom.
\end{proof}

\end{document}